\documentclass{LMCS}
\usepackage{proof}
\usepackage{amssymb}
\usepackage{enumerate,hyperref}

\newcommand{\lpi}{\widehat{\pi}}%------------------------lazy pi
\newcommand{\lapp}[2]{#1\widehat{@}#2}%--------------------lazy append

\newcommand{\lb}{\lambda}

\newcommand{\lbd}{\leftthreetimes\!}%------------------------------fancy lambda

\newcommand{\com}[2]{\langle #1|#2\rangle}%-----------commands
\newcommand{\ctt}[1]{\{#1\}}%---------------coercion command to term

\newcommand{\app}[2]{#1@#2}

\newcommand{\cov}{{*}}%----the special co-variable

\newcommand{\mt}{\tilde{\mu}}%-----------mu-tilde

\newcommand{\ol}{\overline{\lambda} }%----------lambda-bar

\newcommand{\lJmse}{\mathbf{\lambda J}^{\mathbf mse}}%----------lambdaJmse
\newcommand{\lJms}{\mathbf{\lambda J}^{\mathbf ms}}%----------lambdaJms
\newcommand{\lJm}{\mathbf{\lambda J}^{\mathbf m}}%----------lambdaJm
\newcommand{\lm}{\mathbf{\lambda}^{\mathbf m}}%----------lambda-m
\newcommand{\lJ}{\mathbf{\lambda J}}%-------------------lambdaJ

\newcommand{\lJmseS}{\lb 2\mathbf{J}^{\mathbf mse}}%----------lambda2Jmse

\newcommand{\ldfS}{\underline{\lb}2}%---------lambda underbar 2

\newcommand{\lJmseO}{\lb\omega\mathbf{J}^{\mathbf mse}}%----------lambdaOmegaJmse

\newcommand{\ldfO}{\underline{\lb}\omega}%---------lambda underbar Omega

\newcommand{\HJmse}{\lb H\mathbf{J}^{\mathbf mse}}%----------HJmse
\newcommand{\HOL}{\underline{\lb}H}%--------------------------HOL

\newcommand{\lmmt}{\ol\mu\tilde{\mu}}%-----------lambda mu mu-tilde
\newcommand{\ov}[1]{\overline{#1}}

\newcommand{\J}{J}
\newcommand{\m}{m}
\newcommand{\s}{s}
\newcommand{\e}{e}

\newcommand{\sm}[1]{{#1}^{\dagger}}%------------------------Jms to Jm
\newcommand{\es}[1]{{#1}^{\circ}}%------------------------Jmse to Jms

\newcommand{\suc}[1]{\mathsf{s}(#1)}%----------successor

\newcommand{\imp}{\supset}%-----implication, arrow type
\newcommand{\PROP}{PROP}

\newcommand{\seqc}[3]{#1\stackrel{#2}{\longrightarrow}#3}%----sequent for commands
\newcommand{\seql}[4]{#1|#2:#3\vdash #4}%----sequent for lists/co-terms

\newcommand{\D}{\mathcal{D}}
\newcommand{\E}{\mathcal{E}}
\newcommand{\X}{\mathsf{X}}

\def\doi{5 (2:11) 2009}
\lmcsheading%
{\doi}
{1--36}
{}
{}
{Feb.~29, 2008}
{May~\phantom{.}25, 2009}
{}   

\begin{document}

\title[CPS and Strong Normalisation for Intuitionistic Sequent
  Calculi]{Continuation-Passing Style and Strong Normalisation for
  Intuitionistic Sequent Calculi\rsuper*}
\author[J.~Esp\'\i{}rito~Santo]{Jos\'e Esp\'\i{}rito Santo\rsuper a}
\address{{\lsuper{a,c}}Departamento de Matem\'{a}tica, Universidade do Minho,
Portugal}
\email{\{jes,luis\}@math.uminho.pt}
\author[R.~Matthes]{Ralph Matthes\rsuper b}
\address{{\lsuper b}Institut de Recherche en Informatique de Toulouse (IRIT), C.N.R.S. and University of Toulouse III (Paul Sabatier), France}
\author[L.~Pinto]{Lu\'\i{}s Pinto\rsuper c}
%\address{Departamento de Matem\'{a}tica, Universidade do Minho, Portugal}
%\email{luis@math.uminho.pt} 
\titlecomment{{\lsuper*}Extended version of the conference contribution \cite{ESMP:TLCA07} by the same authors}

\begin{abstract}
\noindent
The intuitionistic fragment of the call-by-name version of Curien
and Herbelin's $\lmmt$-calculus is isolated and proved strongly
normalising by means of an embedding into the simply-typed
$\lambda$-calculus. Our embedding is a
continuation-and-garbage-passing style translation, the inspiring
idea coming from Ikeda and Nakazawa's translation of Parigot's
$\lb\mu$-calculus.  The embedding strictly simulates reductions while usual
continuation-passing-style transformations erase permutative
reduction steps. For our intuitionistic sequent calculus, we even
only need ``units of garbage'' to be passed. We apply the same
method to other calculi, namely successive extensions of the
simply-typed $\lb$-calculus leading to our intuitionistic system,
and already for the simplest extension we consider ($\lb$-calculus
with generalised application), this yields the first proof of strong
normalisation through a reduction-preserving embedding. The results
obtained extend to second and higher-order calculi.
\end{abstract}

\maketitle

\tableofcontents

%*********************************************************************************
\section{Introduction}\label{sec:introduction}

CPS (continuation-passing style) translations are a tool with
several theoretical uses. One of them is an interpretation between
languages with different type systems or logical infra-structure,
possibly with corresponding differences at the level of program
constructors and computational behavior. Examples are when the
source language (but not the target language): (i) allows
permutative conversions, possibly related to connectives like
disjunction \cite{Groote02}; (ii) is a language for classical logic,
usually with control operators
\cite{FelleisenFriedmanKohlbeckerDubaLICS1986,GriffinPOPL1990,IkedaNakazawa06};
(iii) is a language for type theory
\cite{BartheHatcliffSoerensen97,BartheHatcliffSoerensen99}
(extending (ii) to variants of pure type systems that have dependent
types and polymorphism).

This article is about CPS translations for intuitionistic sequent
calculi. The source and the target languages will differ neither
in the reduction strategy (they will be both call-by-name) nor at
the types/logic (they will be both based on intuitionistic
implicational logic); instead, they will differ in the structural
format of the type system: the source is in the sequent calculus
format (with cut and left introduction) whereas the target is in
the natural deduction format (with elimination/application). From
a strictly logical point of view, this seems a new
proof-theoretical use for double-negation translations.

Additionally, we insist that our translations strictly simulate reduction.
This is a strong requirement, not present, for instance in the concept
of reflection of \cite{SabryWadlerFP1996}. It seems to have been
intended by \cite{BartheHatcliffSoerensen97}, however does not show up
in the journal version \cite{BartheHatcliffSoerensen99}. But it is,
nevertheless, an eminently useful requirement if one wants to infer
strong normalisation of the source calculus from strong normalisation
of the simply-typed $\lambda$-calculus, as we do. In order to achieve
strict simulation, we define continuation-and-garbage passing style (CGPS)
translations, following an idea due to Ikeda and Nakazawa
\cite{IkedaNakazawa06}. Garbage will provide room for observing
reduction where continuation-passing alone would inevitably produce an
identification, leading to failure of strict simulation in several published
proofs for variants of operationalized classical logic, noted by
\cite{NakazawaTatsuta03} (the problem being $\beta$-reductions under
vacuous $\mu$-abstractions). As opposed to \cite{IkedaNakazawa06}, in
our intuitionistic setting garbage can be reduced to ``units'', and
garbage reduction is simply erasing a garbage unit.

The main system we translate is the intuitionistic fragment of the
call-by-name restriction of the $\lmmt$-calculus
\cite{CurienHerbelinFP2000}, here named $\lJmse$. The elaboration
of this system is interesting on its own. We provide a CPS and a
CGPS translation for $\lJmse$. We also consider other
intuitionistic calculi, whose treatment can be easily derived from
the results for $\lJmse$. Among these is included, for instance,
the $\lambda$-calculus with generalised application. For all these
systems a proof of strong normalisation through a reduction-preserving
embedding into
the simply-typed $\lambda$-calculus is provided for the first
time.

The article is an extended version of the conference contribution
of the same authors \cite{ESMP:TLCA07}.  It is organized as
follows: Section \ref{sec:lambdaJmse} presents $\lJmse$. Section
\ref{sec:other-systems} compares $\lJmse$ with other systems, and
obtains as a by-product confluence of $\lJmse$. Sections
\ref{sec:cgps} deals with the C(G)PS translation of $\lJmse$ and
its subsystems. Section \ref{sec:higher-order-systems} extends the
results to systems $F$, $F^{\omega}$ and intuitionistic
higher-order logic. Section \ref{sec:conclusions} compares this
work with related work and concludes.

%*********************************************************************************

%*********************************************************************************
\section{An intuitionistic sequent calculus}\label{sec:lambdaJmse}

In this section, we define and identify basic properties of the
calculus $\lJmse$. A detailed explanation of the connection
between $\lJmse$ and $\lmmt$ is left to the next section.

There are three classes of expressions in $\lJmse$:
$$
\begin{array}{lcrcl}
\textrm{(Terms)} &  & t,u,v & ::= & x\,|\,\lambda x.t\,|\,\ctt c\\
\textrm{(Co-terms)} &  & l & ::= & []\,|\,u::l\,|\,(x)c\\
\textrm{(Commands)} &  & c & ::= & tl
\end{array}
$$
Terms can be variables (of which we assume a denumerable set
ranged over by letters $x$, $y$, $w$, $z$), lambda-abstractions
$\lambda x.t$ or coercions $\ctt c$ from commands to
terms\footnote{A version of $\lJmse$ with implicit coercions would
be possible but to the detriment of the clarity, in particular, of
the reduction rule $\epsilon$ below.}. A \emph{value} is a term which is
either  a variable or a lambda-abstraction. We use letter $V$ to
range over values.

Co-terms provide means of forming lists of arguments, generalised
arguments \cite{JoachimskiMatthesRTA2000}, or explicit
substitutions. A co-term of the form $(x)c$ binds variable $x$ in
$c$ and provides the generalised application facility. Operationally
it can be thought of as ``substitute for $x$ in $c$''. A co-term of
the form $[]$ or $u::l$ is called an \emph{evaluation context} and
is denoted by $E$. An evaluation context of the form $u::l$ allows
for multiary applications, and when passed to a term it indicates
that after consumption of argument $u$ computation should carry on
with arguments in $l$. $[]$ marks the end of an evaluation context
and compensates the impossibility of writing $(x)x$.

A command
$tl$ has a double role: if $l$ is of the form $(x)c$, $tl$ is an
explicit substitution; otherwise, $tl$ is a general form of
application.

In writing expressions, sometimes we add parentheses to help their
parsing. Also, we assume that the scope of binders $\lambda x$ and
$(x)$ extends as far as possible. Usually we write only one $\lb$
for multiple abstraction.

In what follows, we reserve letter $T$ (``term in a large sense'')
for arbitrary expressions. We write $x\notin T$ if $x$ does not
occur free in $T$. Substitution $[t/x]T$ of a term $t$ for all free
occurrences of a variable $x$ in $T$ is defined as expected, where
it is understood that bound variables are chosen so that no variable
capture occurs.
\[
\begin{array}{lcl@{\qquad\qquad}lcl}
[t/x]x&=&t&[t/x][]&=&[]\\
\relax[t/x]y&=&y\textrm{ if $x\neq y$}&[t/x](u::l)&=&[t/x]u::[t/x]l\\
\relax[t/x](\lb y.u)&=&\lb y.[t/x]u&[t/x]((y)c)&=&(y)[t/x]c\\
\relax[t/x]\ctt c&=&\ctt{[t/x]c}&[t/x](ul)&=&[t/x]u[t/x]l\\
\end{array}
\]
Evidently, syntactic classes are respected by substitution, i.\,e., $[t/x]u$ is a term, $[t/x]l$ is
a co-term and $[t/x]c$ is a command.\\

The calculus $\lJmse$ has a form of sequent for each class of expressions:
$$\Gamma\vdash t:A\qquad \quad \Gamma|l:A\vdash B\qquad \quad \seqc{\Gamma}{c}{B}$$

Letters $A,B,C$ are used to range over the set of types
(=formulas), built from a base set of type variables (ranged over
by $X$) using the function type (that we write $A\supset B$). In
sequents, contexts $\Gamma$ are viewed as finite sets of
declarations $x:A$, where no variable $x$ occurs twice. The
context $\Gamma,x:A$ is obtained from $\Gamma$ by adding the
declaration $x:A$, and will only be written if this yields again a
valid context, i.\,e., if $x$ is not declared in $\Gamma$.
Similarly, $\Gamma,\Delta$ is the union of $\Gamma$ and $\Delta$,
and assumes that the sets of variables declared in $\Gamma$ and
$\Delta$ are disjoint. We can think of a term (resp. co-term) as
an annotation for a selected formula in the \emph{rhs} (resp.
\emph{lhs}). Commands annotate sequents generated as a result of
logical cuts, where there is no selected formula on the \emph{rhs}
or \emph{lhs}; as such we write them on top of the sequent arrow.

\smallskip
The typing rules of $\lJmse$ are presented in Figure~\ref{fig:ljmse},
stressing the parallel between left and right rules.
\begin{figure}\caption{Typing rules of $\lJmse$}\label{fig:ljmse}$$
\begin{array}{c}
\infer[LAx]{\seql{\Gamma}{[]}{A}{A}}{}\quad\quad\infer[RAx]{\Gamma,x:A\vdash x:A}{}\\ \\
\infer[LIntro]{\seql{\Gamma}{u::l}{A\imp B}{ C}}{\Gamma\vdash u:A&
\seql{\Gamma}{l}{B}{C}}\quad\quad
\infer[RIntro]{\Gamma\vdash\lambda x.t:A\imp B}{\Gamma,x:A\vdash
t:B}\\
\\
\infer[LSel]{\seql{\Gamma}{(x)c}{A}{B}}{\seqc{\Gamma,x:A}{c}{B}}
\quad\quad \infer[RSel]{\Gamma\vdash
\ctt c:A}{\seqc{\Gamma}{c}{A}}\\
\\
\infer[Cut]{\seqc{\Gamma}{tl}{B}}{\Gamma\vdash t:A &
\seql{\Gamma}{l}{A}{B}}
\end{array}
$$
\end{figure}

The following
other forms of cut are admissible  as typing rules for
substitution for each class of expressions:
\[
\begin{array}{c}
\infer{\Gamma\vdash[t/x]u:B}{\Gamma\vdash t:A&\Gamma,x:A\vdash
u:B}\quad\quad
\infer{\seql{\Gamma}{[t/x]l}{B}{C}}{\Gamma\vdash t:A&\seql{\Gamma,x:A}{l}{B}{C}}\\[2ex]
\infer{\seqc{\Gamma}{[t/x]c}{B}}{\Gamma\vdash
t:A&\seqc{\Gamma,x:A}cB}
\end{array}
\]
We also have the usual weakening rules: If a sequent with context
$\Gamma$ is derivable and $\Gamma$ is replaced by a context $\Gamma'$
that is a superset of $\Gamma$, then also this sequent is derivable.

We consider the following base reduction rules on expressions:

$$
\begin{array}{rrclrrcl}
(\beta) & (\lambda x.t)(u::l)& \rightarrow &
u((x)tl)&\quad\quad(\mu) & (x)xl & \rightarrow & l,\textrm{ if
$x\notin
l$}\\
(\pi) & \ctt{tl}E & \rightarrow & t\,(\app lE)&
(\epsilon) & \ctt{t[]} & \rightarrow & t\\
(\sigma) & t(x)c & \rightarrow & [t/x]c,\\
\end{array}
$$
where, in general, $\app l{l'}$ is a co-term that represents an
``eager'' concatenation of $l$ and $l'$, viewed as lists, and is
defined as follows\footnote{Concatenation is ``eager'' in the
sense that, in the last case, the right-hand side is not
$(x)\ctt{tl}l'$ but, in the only important case that $l'$ is an
evaluation context $E$, its $\pi$-reduct. One immediately verifies
$\app l{[]}=l$ and $\app{(\app l{l'})}{l''}=\app l{(\app
{l'}{l''})}$ by induction on $l$. Associativity would not hold with
the lazy version of $\app{}{}$. Nevertheless, one would get that
the respective left-hand side reduces in at most one $\pi$-step to
the right-hand side.}:
$$\app{[]}{l'}=l'\qquad
\app{(u::l)}{l'}=u::(\app l{l'})\qquad \app{((x)tl)}{l'}=(x)t\,(\app
l{l'})$$
The one-step reduction relation $\rightarrow$ is
inductively defined as the term closure of the reduction rules, by
adding the following closure rules to the above initial cases of
$\rightarrow$:
$$
\begin{array}{lcl}
t \to t' &\Longrightarrow& \lambda x. t \to \lambda x. t',\quad tl \to t'l,\quad t::l \to t'::l,\\
l \to l' &\Longrightarrow& u::l \to u::l',\quad tl \to tl',\\
c \to c' &\Longrightarrow& (x)c \to (x)c',\quad \ctt c\to \ctt {c'}.
\end{array}
$$

The reduction rules $\beta$, $\pi$ and $\sigma$ are relations on
commands. The reduction rule $\mu$ (resp. $\epsilon$) is a
relation on co-terms (resp. terms). Rules $\beta$ and $\sigma$
generate and execute an explicit substitution, respectively. Rule
$\pi$ appends fragmented co-terms, bringing the term $t$ of the
$\pi$-redex $\ctt{tl}E$ closer to root position. Also, notice here
the restricted form of the outer co-term $E$. This restriction
characterizes call-by-name reduction \cite{CurienHerbelinFP2000}.
A $\mu$-reduction step that is not at the root
has necessarily one of two forms: (i)
$t(x)xl \rightarrow tl$, which is the execution of a linear
substitution; (ii) $u::(x)xl\rightarrow u::l$, which is the
simplification of a generalised argument. Rule $\mu$ undoes
the sequence of inference steps consisting in un-selecting a
formula and giving it the name $x$, followed by immediate
selection of the same formula. The proviso $x\notin l$ guarantees
that no contraction was involved. Finally, rule $\epsilon$ erases
an empty list under $\ctt{\_}$. Notice that empty lists are
important under $(x)$. Another view of $\epsilon$ is as a way of
undoing a sequence of two coercions: the ``coercion'' of a term
$t$ to a command $t[]$, immediately followed by coercion to a term
$\ctt{t[]}$. By the way, $\ctt{c}[]\rightarrow c$ is a
$\pi$-reduction step. Most of these rules have genealogy: see
Section \ref{subsec:spectrum}.

The $\beta\pi\sigma$-normal forms are obtained by constraining
commands to one of the two forms $V[]$ or $x(u::l)$, where $V,u,l$
are $\beta\pi\sigma$-normal values, terms and co-terms
respectively. The $\beta\pi\sigma\epsilon$-normal forms are
obtained by requiring additionally that, in coercions $\ctt{c}$,
$c$ is of the form $x(u::l)$ (where $u,l$ are
$\beta\pi\sigma\epsilon$-normal terms and co-terms respectively).
$\beta\pi\sigma\epsilon$-normal forms correspond to the multiary
normal forms of \cite{SchwichtenbergTCS1999}. If we further impose
$\mu$-normality as in \cite{SchwichtenbergTCS1999}, then co-terms
of the form $(x)x(u::l)$ obey to the additional restriction that
$x$ occurs either in $u$ or $l$.\\

Subject reduction holds for $\rightarrow$, i.\,e., the
following rules are admissible:
\[
\infer{\Gamma\vdash t':A}{\Gamma\vdash t:A&t\to t'}\quad
\infer{\Gamma|l':A\vdash B}{\Gamma|l:A\vdash B&l \to l'}\quad
\infer{\seqc{\Gamma}{c'}B}{\seqc\Gamma cB & c\to c'}
\]
This fact is established with the help of the admissible rules for
typing substitution and  with the help of yet another admissible
form of cut for typing the append operator:
\[
\infer{\seql{\Gamma}{\app l{l'}}{A}{C}}{\seql{\Gamma}{l}{A}{
B}&\seql{\Gamma}{l'}{B}{C}}
\]

We offer now a brief analysis of critical pairs in $\lJmse$
\footnote{For higher-order rewrite systems, see the formal
definition in \cite{MayrNipkow98}.}.

There is a self-overlap of $\pi$ ($\ctt{\ctt{tl}E'}E$), there are
overlaps between $\pi$ and any of $\beta$ ($\ctt{(\lb x.t)(u::l)}E$),
$\sigma$ ($\ctt{t(x)c}E$) and $\epsilon$ (the latter in two different
ways from $\ctt{t[]}E$ and $\ctt{\ctt{tl}[]}$).  Finally, $\mu$
overlaps with $\sigma$ in two different ways from $t(x)xl$ for
$x\notin l$ and $(x)(x(y)c)$ for $x\notin c$.  The last four critical
pairs are trivial in the sense that both reducts are identical. Also
the other critical pairs are joinable in the sense that both terms
have a common $\to^*$-reduct. We only show this for the first case:
$\ctt{tl}E'\to t(\app lE')$ by $\pi$, hence also
$$\ctt{\ctt{tl}E'}E\to\ctt{t(\app lE')}E=:L.$$
On the other hand, a direct application of $\pi$ yields
$$\ctt{\ctt{tl}E'}E\to\ctt{tl}(\app {E'}{E})=:R.$$
 Thus the critical
pair consists of the terms $L$ and $R$. $L\to t(\app{(\app
l{E'})}{E})$ and $R\to t(\app{l}{(\app {E'}{E})})$, hence $L$ and $R$
are joinable by associativity of $\app{}{}$.

We remark that the first three critical pairs (like the one just
shown) are of a particularly simple nature: The forking term is of the
form $\ctt{c}E$ with $c$ any of the command redexes, i.\,e., a
left-hand side of $\beta$, $\pi$ or $\sigma$. The $L$ term is obtained
by reducing $c$ to the respective right-hand side $c'$ of that rule,
and the $R$ term comes from applying $\pi$ at the root. Since $c'$ is
again a command, $L=\ctt{c'}E$ can be reduced by $\pi$ to a term
$L'$. The decisive feature of $@$ is that $R\to L'$ by an instance
of the rule $c\to c'$ where the co-term part $l$ of $c=tl$ is replaced
by $\app lE$.

Since the critical pairs are joinable, the relation $\to$ is
locally confluent \cite{MayrNipkow98}. Thus, from
Corollary~\ref{cor:SN-for-LambdaJmse} below and Newman's Lemma,
$\to$ is confluent on typable terms.
Confluence on all terms is proved in the next section.

%*********************************************************************************

%+++++++++++++++++++++++++++++++++++++++++++++++++++++++

\section{Comparison with other systems}\label{sec:other-systems}

In this section we show that $\lJmse$ can be generated ``from
above'' - being the intuitionistic fragment of the call-by-name
restriction of Curien and Herbelin's $\lmmt$-calculus; and ``from
below'' - being the end-point of a spectrum of successively more
general intuitionistic systems, starting from the ordinary
$\lambda$-calculus. This latter result, by showing that the
systems in the spectrum are subsystems of $\lJmse$, will allow us
to adapt easily the result about $\lJmse$ to (new) results about
its subsystems. In addition, we will obtain, as a by-product,
a proof of confluence for $\lJmse$ even for the untypable terms.

\subsection{$\lJmse$ as the intuitionistic fragment of CBN
$\lmmt$}\label{subsec:int-fragment}

After a recapitulation of a call-by-name version of
$\lmmt$-calculus, we restrict it to the intuitionistic case and
rediscover $\lJmse$.

\subsubsection{The call-by-name $\lmmt$-calculus}\label{subsec:lambda-mu-mu-tilde}

Here, we recall the Curien and Herbelin's
$\lmmt$-cal\-cu\-lus \cite{CurienHerbelinFP2000}. More precisely, we
only consider implication (i.\,e., we do not include the
subtraction connective) and we present the call-by-name restriction
of the system.

Expressions are either terms, co-terms or
commands and are defined by
the following grammar:
$$
\begin{array}{rclcrclcrcl}
 t,u,v & ::= & x\,|\,\lambda x.t\,|\,\mu a.c&\quad&
 e & ::= & a\,|\,u::e\,|\,\mt x.c&\quad&
 c & ::= & \com{t}{e}
\end{array}
$$

\noindent Variables (resp.~co-variables) are ranged over by $x$,
$y$, $z$ (resp.~$a$, $b$, $c$). An \emph{evaluation context} $E$
is a co-term of the form $a$ or $u::e$.

There is one kind of sequent per each syntactic class

$$\Gamma\vdash t:A|\Delta \quad \quad \Gamma|e:A\vdash \Delta
\quad \quad c:(\Gamma\vdash\Delta)$$

\noindent Typing rules are given in Figure~\ref{fig:cbnlmmt}.
\begin{figure}\caption{Typing rules of CBN $\lmmt$}\label{fig:cbnlmmt}
$$
\begin{array}{c}
\infer[LAx]{\Gamma|a:A\vdash a:A,\Delta}{}\quad\quad\infer[RAx]{\Gamma,x:A\vdash x:A|\Delta}{} \\ \\
\infer[LIntro]{\Gamma|u::e:A\supset B\vdash \Delta}{\Gamma\vdash
u:A|\Delta & \Gamma|e:B\vdash
\Delta}\quad\quad\infer[RIntro]{\Gamma\vdash\lambda x.t:A\supset
B|\Delta}{\Gamma,x:A\vdash t:B|\Delta}
\\
\\
\infer[LSel]{\Gamma|\mt x.c:A\vdash
\Delta}{c:(\Gamma,x:A\vdash\Delta)}\quad\quad
\infer[RSel]{\Gamma\vdash \mu a.c:A|\Delta}{c:(\Gamma\vdash
a:A,\Delta)}\\
\\
\infer[Cut]{\com{t}{e}:(\Gamma\vdash\Delta)}{\Gamma\vdash
t:A|\Delta & \Gamma|e:A\vdash \Delta}
\end{array}
$$
\end{figure}

There are 6 substitution operations altogether:
$$[t/x]c\quad\quad[t/x]u\quad\quad[t/x]e\quad\quad[e/a]c\quad\quad[e/a]u\quad\quad[e/a]e'$$

\noindent
We consider 5 reduction rules:

$$
\begin{array}{rrclcrrcl}
(\beta) & \com{\lambda x.t}{u::e}& \rightarrow & \com{u}{\mt
x.\com{t}{e}}&\quad&(\mu) & \mt x.\com{x}{e} & \rightarrow &
e,\textrm{ if $x\notin
e$}\\
(\pi) & \com{\mu a.c}{E} & \rightarrow & [E/a]c&\quad&(\mt) & \mu
a.\com{t}{a} & \rightarrow & t,\textrm{ if $a\notin t$}\\
(\sigma) & \com{t}{\mt x.c} & \rightarrow & [t/x]c&&&&&
\end{array}
$$

\noindent These are the reductions considered by Polonovski in
\cite{PolonovskiFOSSACS2004}, with three provisos. First, the
$\beta$-rule for the subtraction connective is not included.
Second, in the $\pi$-rule, the co-term involved is an evaluation
context $E$; this is exactly what characterizes the call-by-name
restriction of $\lmmt$ \cite{CurienHerbelinFP2000}. Third, the
naming of the rules is non-standard. Curien and Herbelin (and
Polonovski as well) name rules $\pi$ and $\sigma$ as $\mu$, $\mt$,
respectively. The name $\mu$ has moved to the rule called $se$ in
\cite{PolonovskiFOSSACS2004}. By symmetry, the rule called $sv$ by
Polonovski is now called $\mt$. The reason for this change is
explained below by the spectrum of systems in Section
\ref{subsec:spectrum}: the rule we now call $\pi$ (resp. $\mu$) is
the most general form of the rule with the same name in the system
$\lJ$ (resp. $\lJm$), and therefore its name goes back to
\cite{JoachimskiMatthesRTA2000} (resp. \cite{jesLuisTLCA03},
actually back to \cite{SchwichtenbergTCS1999}).

%++++++++++++++++++++++++++++++++++++++++++++++++++++++++++++++++++++++
\subsubsection{The intuitionistic fragment of CBN
$\lmmt$}\label{subsubsec:int-fragment}

The following description is in the style of Section 2.13 of Herbelin's
habilitation thesis \cite{HerbelinHabilitation}.

Let $\cov$ be a fixed co-variable. The intuitionistic terms,
co-terms and commands are generated by the grammar

$$
\begin{array}{lcrcl}
\textrm{(Terms)} &  & t,u,v & ::= & x\,|\,\lambda x.t\,|\,\mu\cov.c\\
\textrm{(Co-terms)} &  & e & ::= & \cov\,|\,u::e\,|\,\mt x.c\\
\textrm{(Commands)} &  & c & ::= & \com{t}{e}
\end{array}
$$

\noindent Terms have no free occurrences of co-variables. Each
co-term or command has exactly one free occurrence of $\cov$.
Sequents are restricted to have exactly one formula in the RHS.
Therefore, they have the particular forms $\Gamma\vdash t:A$, $
\Gamma|e:A\vdash \cov:B$ and $c:(\Gamma\vdash\cov:B)$. We omit
writing the intuitionistic typing rules. Reduction rules read as
for $\lmmt$, except for $\pi$ and $\mt$:
$$(\pi) \quad \com{\mu\cov.c}{E} \rightarrow  [E/\cov]c\qquad\quad
(\mt) \quad \mu\cov.\com{t}{\cov} \rightarrow t$$

\noindent Since $\cov\notin t$, $[E/\cov]t=t$. Let us spell out
$[E/\cov]c$ and $[E/\cov]e$.

$$
\begin{array}{rclcrcl}
[E/\cov]\com{t}{e}&=&\com{t}{[E/\cov]e}&\quad&{[}E/\cov{]}(u::e)&=&u::[E/\cov]e\\
{[}E/\cov{]}\cov&=&E&\quad&{[}E/\cov{]}(\mt x.c)&=&\mt x.[E/\cov]c\\
\end{array}
$$

\noindent If we define rule $\pi$ as $\com{\mu
\cov.\com{t}{e}}{E}\rightarrow\com{t}{[E/\cov]e}$ and
${[}E/\cov{]}(\mt x.\com{t}{e})=\mt x.\com{t}{[E/\cov]e}$ we can
avoid using $[E/\cov]c$ altogether.

The $\lJmse$-calculus is obtained from the intuitionistic fragment
as a mere notational variant. The co-variable $\cov$ disappears
from the syntax. The co-term $\cov$ is written $[]$. $\ctt{c}$ is
the coercion of a command to a term, corresponding to $\mu
\cov.c$. This coercion is what remains of the $\mu$ binder in the
intuitionistic fragment. Since there is no $\mu$, there is little
sense for the notation $\mt$. So we write $(x)c$ instead of $\mt
x.c$. Reduction rule $\mt$ now reads $\ctt{t[]}\rightarrow t$ and
is renamed as $\epsilon$. Sequents $\Gamma|e:A\vdash \cov:B$ and
$c:(\Gamma\vdash\cov:B)$ are written $\seql{\Gamma}{e}{A}{B}$ and
$\seqc{\Gamma}{c}{B}$. Co-terms are ranged over by $l$ (instead of
$e$) and thought of as generalised lists. Finally, $[E/\cov]l$ is
written $\app lE$.

%*********************************************************************************

\subsection{A spectrum of intuitionistic calculi}\label{subsec:spectrum}

The calculus $\lJmse$ can also be explained as the end product of
successive extensions of the simply-typed $\lambda$-calculus through
several intuitionistic calculi, as illustrated in
Fig.~\ref{fig:spectrum}, which includes both natural deduction
systems and sequent calculi other than $\lJmse$.

\newcommand{\lTo}[1]{\stackrel{#1}{\longleftarrow}}

%----------------------------------------------------------------
\begin{figure}\caption{A spectrum of intuitionistic
calculi}\label{fig:spectrum}
\[
\begin{array}{|ccc|} \hline
&&\\
&\lJmse\lTo{\e}\lJms\lTo{\s}\lJm\lTo{\m}\lJ\lTo{\J}\lambda&\\
&&\\
&\begin{tabular}{lr} Sequent
Calculus\quad\quad\quad\quad\quad&\quad\quad\quad\quad\quad Natural
Deduction
\end{tabular}&\\
\hline
\end{array}
\]
\end{figure}

%----------------------------------------------------------------

Each extension step adds both a new feature and a reduction rule
to the preceding calculus. The following table summarizes these
extensions.

$$
\begin{array}{l|l|l}
\textrm{calculus\,}&\textrm{\,reduction rules\,}&\textrm{\,feature added}\\
\hline \lb&\beta\\
\lJ&\beta,\pi&\textrm{generalised application}\\
\lJm&\beta,\pi,\mu&\textrm{multiarity}\\
\lJms&\beta,\pi,\mu,\sigma&\textrm{explicit substitution}\\
\lJmse&\beta,\pi,\mu,\sigma,\epsilon&\textrm{empty lists of
arguments}
\end{array}
$$

\noindent The scheme for naming systems and reduction rules
intends to be systematic (and in particular explains the name
$\lJmse$).

The path between the two end-points of this spectrum visits and
organizes systems known from the literature. $\lJ$ is a variant of
the calculus $\Lambda $J of \cite{JoachimskiMatthesRTA2000}.
$\lJm$ is a variant of the system in \cite{jesLuisTLCA03}.
$\lJmse$ is studied in \cite{jesTLCA07} under the name
$\lambda^{Gtz}$. This path is by no means unique. Other
intermediate systems could have been visited (like the multiary
$\lambda$-calculus $\lm$, named $\lambda Ph$ in
\cite{jesLuisTLCA03}), had the route been a different one, i.\,e.,
had the different new features been added in a different order.
The reader is referred to the literature for the specific
motivations underlying the introduction of the intermediate
systems $\lJ$, $\lJm$, and $\lJms$. Here, their interest lies in
being the successive systems obtained by the addition, in a
specific order, of the features exhibited by $\lJmse$.

Each system $\mathcal{L}\in\{\lJ,\lJm,\lJms\}$ embeds in the system
immediately after it in this spectrum, in the sense of allowing a
mapping that strictly simulates reduction. Hence, strong
normalisation is inherited from $\lJmse$ all the way down to $\lJ$.
Also, each $\mathcal{L}\in\{\lJ,\lJm,\lJms\}$ has, by composition,
an embedding $g_{\mathcal{L}}$ in $\lJmse$. Let us see all this with
some detail.

%++++++++++++++++++++++++++++++++++++++++++++++++++++++++++++++++++

\subsubsection{$\lJ$-calculus.} The terms of $\lJ$ are generated by the
grammar:

$$
\begin{array}{lcrcl}
 &  & t,u,v & ::= & x\,|\,\lambda x.t\,|\,t(u,x.v)\\
\end{array}
$$

\noindent Construction $t(u,x.v)$ is called generalised application.
Following \cite{JoachimskiMatthesRTA2000}, $(u,x.v)$ is called a
generalised argument; they will be denoted by the letters $R$ and
$S$. Typing rules for $x$ and $\lambda x.t$ are as usual, and the
new rule is that of generalised application, given in
Figure~\ref{fig:lj}.
\begin{figure}\caption{Typing rules of $\lJ$}\label{fig:lj}
\[
\begin{array}{c}
\infer[Ax]{\Gamma,x:A\vdash x:A}{}\quad\quad
\infer[Intro]{\Gamma\vdash\lambda x.t:A\imp B}{\Gamma,x:A\vdash
t:B}\\
\\\infer[GApp]{\Gamma\vdash t(u,x.v):C}{\Gamma\vdash
t:A\imp B & \Gamma\vdash u:A &\Gamma,x:B\vdash v:C}
\end{array}
\]
\end{figure}
% ----------------------------

Reduction rules are as in \cite{JoachimskiMatthesRTA2000}, except
that $\pi$ is defined in the ``eager'' way:
$$
\begin{array}{rrclcrrcl}
(\beta) & (\lambda x.t)(u,y.v)& \rightarrow & [[u/x]t/y]v&\quad&(\pi) & tRS & \rightarrow & t(\app RS)\\
\end{array}
$$
where the generalised argument $R@S$ is defined by
recursion on $R$:
$$
\begin{array}{rclcrcl}
\app{(u,x.V)}S&=&(u,x.VS)&\quad&\app{(u,x.tR')}S&=&(u,x.t(\app{R'}S)),\\
\end{array}
$$
for $V$ a value, i.\,e., a variable or a
$\lambda$-abstraction. The operation $\app{}{}$ is associative,
which allows to join the critical pair of $\pi$ with itself as
before for $\lJmse$. The other critical pair stems from the
interaction of $\beta$ and $\pi$ and is joinable as well.

Strong normalisation of typable terms immediately follows from that
of $\Lambda$J in \cite{JoachimskiMatthes03}, but in the present
article, we even get an embedding into $\lb$.

\medskip
Although we won't use it, we recall the embedding
$\J:\lb\rightarrow\lJ$ just for completeness:

$$
\begin{array}{rcl}
\J(x)&=&x\\
\J(\lb x.t)&=&\lb x.\J(t)\\
\J(tu)&=&\J(t)(\J(u),x.x)
\end{array}
$$

%++++++++++++++++++++++++++++++++++++++++++++++++++++++++++++++++++
\subsubsection{$\lJm$-calculus}

We offer now a new, lighter,
presentation of the system in \cite{jesLuisTLCA03}. The
expressions of $\lJm$ are given by the grammar:
$$
\begin{array}{lrclclrcl}
\textrm{(Terms)} & t,u,v & ::= & x\,|\,\lambda
x.t\,|\,t(u,l)&&
\textrm{(Co-terms)} & l & ::= & u::l\,|\,(x)v\\
\end{array}
$$
The application $t(u,l)$ is both generalised and multiary.
Multiarity is the ability of forming a chain of arguments, as in
$t(u_1,u_2::u_3::(x)v)$. By the way, this term is written
$t(u_1,u_2::u_3::[],(x)v)$ in the syntax of \cite{jesLuisTLCA03}.
There are two kinds of sequents: $\Gamma\vdash t:A$ and
$\Gamma|l:A\vdash B$. Typing rules are given in
Figure~\ref{fig:ljm}.
\begin{figure}\caption{Typing rules of $\lJm$}\label{fig:ljm}
$$
\begin{array}{c}
\infer[Ax]{\Gamma,x:A\vdash
x:A}{}\quad\quad\infer[GMApp]{\Gamma\vdash t(u,l):C}{\Gamma\vdash
t:A\imp B & \Gamma\vdash u:A &\seql{\Gamma}{l}{B}{C}}\\
\\ \infer[LIntro]{\seql{\Gamma}{u::l}{A\imp B}{
C}}{\Gamma\vdash u:A &
\Gamma|l:B\vdash C}\quad\quad
\infer[RIntro]{\Gamma\vdash\lambda x.t:A\imp B}{\Gamma,x:A\vdash
t:B}\\
\\ \infer[Sel]{\seql{\Gamma}{(x)v}{A}{B}}{\Gamma,x:A\vdash v:B}
\end{array}
$$
\end{figure}

We re-define reduction rules of \cite{jesLuisTLCA03} in this new
syntax. Rule $\mu$ can now be defined as a relation on co-terms.
Rule $\pi$ is changed to the ``eager'' version, using letters $R$
and $S$ for generalised arguments, i.\,e., elements of the form
$(u,l)$.

$$
\begin{array}{rrcl}
(\beta_1) & (\lambda x.t)(u,(y)v)& \rightarrow &
[[u/x]t/y]v\\
(\beta_2) & (\lambda x.t)(u,v::l)& \rightarrow &
([u/x]t)(v,l)\\
(\pi) & tRS & \rightarrow & t(\app RS)\\
(\mu) &  (x)x(u,l) & \rightarrow &
u::l,\textrm{ if $x\notin u,l$}\\
\end{array}
$$

\noindent $\beta=\beta_1\cup\beta_2$. The generalised argument $\app
RS$ is defined with the auxiliary notion of the co-term $\app lS$
that is defined by recursion on $l$ by
$$
\begin{array}{rcl}
\app{(u::l)}S&=&u::(\app lS)\\
 \app{((x)V)}S&=&(x)VS,\quad\textrm{for $V$
a value}\\
\app{((x)t(u,l))}S&=&(x)t(u,\app lS)\\
\end{array}
$$
Then, define $\app RS$ by $\app{(u,l)}S=(u,\app lS)$. Since the
auxiliary operation $\app{}{}$ can be proven associative, this also
holds for the operation $\app{}{}$ on generalised arguments. Apart
from the usual self-overlapping of $\pi$ that is joinable by
associativity of $\app{}{}$, there are critical pairs between
$\beta_i$ and $\pi$ that are joinable. The last critical pair is
between $\beta_1$ and $\mu$ and needs a $\beta_2$-step to be joined.

The embedding $\m:\lJ\rightarrow\lJm$ is given by
$$
\begin{array}{rcl}
\m(x)&=&x\\
\m(\lb x.t)&=&\lb x.\m(t)\\
\m(t(u,x.v))&=&\m(t)(\m(u),(x)\m(v))
\end{array}
$$

%++++++++++++++++++++++++++++++++++++++++++++++++++++++++++++++++++
\subsubsection{$\lJms$-calculus}

The expressions of $\lJms$ are given
by:
$$
\begin{array}{lcrclclcrcl}
\textrm{(Terms)} &  & t,u,v & ::= & x\,|\,\lambda
x.t\,|\,tl&\quad&\textrm{(Co-terms)} &  & l & ::= & u::l\,|\,(x)v\\
\end{array}
$$

\noindent The construction $tl$ has a double role: either it is a
generalised and multiary application $t(u::l)$ or it is an explicit
substitution $t(x)v$. See Figure~\ref{fig:ljms} for the typing
rules.
\begin{figure}\caption{Typing rules of $\lJms$: $GMApp$ of $\lJm$ is generalized to $Cut$}\label{fig:ljms}
$$
\begin{array}{c}
\infer[Ax]{\Gamma,x:A\vdash
x:A}{}\quad\quad\infer[Cut]{\Gamma\vdash tl:B}{\Gamma\vdash t:A &
\seql{\Gamma}{l}{A}{B}}\\
\\
 \infer[LIntro]{\seql{\Gamma}{u::l}{A\imp B}{
C}}{\Gamma\vdash u:A &
\seql{\Gamma}{l}{B}{C}}\quad\quad\infer[RIntro]{\Gamma\vdash\lambda x.t:A\imp B}{\Gamma,x:A\vdash
t:B}\\
\\ \infer[Sel]{\seql{\Gamma}{(x)v}{A}{B}}{\Gamma,x:A\vdash v:B}
\end{array}
$$
\end{figure}

The reduction rules are as follows:
$$
\begin{array}{rrclcrrcl}
(\beta) & (\lambda x.t)(u::l)& \rightarrow & u((x)tl)&\quad&(\sigma) & t(x)v & \rightarrow & [t/x]v\\
(\pi) & (tl)(u::l') & \rightarrow & t\,(\app l{(u::l')})&\quad&(\mu)
&  (x)xl & \rightarrow & l,\textrm{ if $x\notin
l$}\\
\end{array}
$$
where the co-term $\app l{l'}$ is defined by
$$
\begin{array}{rcl}
\app{(u::l)}{l'}&=&u::(\app l{l'})\\
\app{((x)V)}{l'}&=&(x)Vl', \quad\textrm{for $V$ a value}\\
\app{((x)tl)}{l'}&=&(x)t\,(\app
l{l'})
\end{array}
$$
Again, $\app{}{}$ is associative and guarantees
the joinability of the critical pair of $\pi$ with itself. The
critical pairs between $\beta$ and $\pi$ and between $\sigma$ and
$\mu$ are joinable as for $\lJmse$. The overlap between $\sigma$ and
$\pi$ is bigger than in $\lJmse$ since the divergence arises for
$t((x)v)(u::l)$ with $v$ an arbitrary term whereas in $\lJmse$,
there is only a command at that place. Joinability is nevertheless
easily established.

Comparing these reduction rules with those of $\lJm$, there is only
one $\beta$-rule, whose effect is to generate a substitution. There
is a separate rule $\sigma$ for substitution execution. The
embedding $\s:\lJm\rightarrow\lJms$ is defined by
$$
\begin{array}{rclcrcl}
\s(x)&=&x&\quad&\s(u::l)&=&\s(u)::\s(l)\\
\s(\lb x.t)&=&\lb x.\s(t)&\quad&\s((x)v)&=&(x)\s(v)\\
\s(t(u,l))&=&\s(t)(\s(u)::\s(l))&&&&\\
\end{array}
$$

Finally, let us compare $\lJms$ and $\lJmse$. In the former, any
term can be in the scope of a selection $(x)$, whereas in the latter
the scope of a selection is a command. But in the latter we have a
new form of co-term $[]$. Since in $\lJmse$ we can coerce any term
$t$ to a command $t[]$, we can translate $\lJms$ into $\lJmse$, by
defining $e((x)t)=(x)e(t)[]$. In fact, one has to refine this idea
in order to get strict simulation of reduction. The embedding
$\e:\lJms\rightarrow\lJmse$ is defined as
$$
\begin{array}{rclcrcl}
\e(x)&=&x&\quad&\e(u::l)&=&\e(u)::\e(l)\\
\e(\lb x.t)&=&\lb x.\e(t)&\quad&\e((x)V)&=&(x)\e(V)[]\\
\e(tl)&=&\ctt{\e(t)\e(l)}&\quad&\e((x)tl)&=&(x)\e(t)\e(l)\\
\end{array}
$$

%-------------------------------------------------
\begin{prop}\label{prop:simulation-by-embeddings}
Each of the embeddings $m$, $s$ and $e$ preserves typability and types and strictly simulates reduction.
\end{prop}

\begin{proof} Preservation of typability and types is immediate by induction on typing derivations. For strict simulation, we  prove by induction
\begin{description}
\item[(i)]
$t\rightarrow t' \Longrightarrow m(t) \rightarrow^+ m(t')$, for any
$t,t'\in\lJ$
\item[(ii)]
$t\rightarrow t' \Longrightarrow s(t) \rightarrow^+ s(t')$, for any
$t,t'\in\lJm$
\item[(iii)]
$t\rightarrow t' \Longrightarrow e(t) \rightarrow^+ e(t')$ and
$e((x)t) \rightarrow^+ e((x)t')$, for any $t,t'\in\lJms$
\end{description}
which for $f\in\{s,e\}$ requires simultaneous proof of:
$l\rightarrow l' \Longrightarrow f(l) \rightarrow^+ f(l').$ We show
only some details of the proof of (iii). (The other statements have
simpler proofs.) In the cases where  $t\rightarrow_R t'$ (resp.
$l\rightarrow l'$), with $R\in\{\beta,\pi,\mu\}$, in $\lJms$, we
have $e(t) \rightarrow_R e(t')$ and $e((x)t) \rightarrow_R e((x)t')$
(resp. $e(l) \rightarrow_R e(l')$) in $\lJmse$. The proof relative
to $\pi$-steps requires commutation of the embedding with the append
operator, that is requires the identity: $e(\app
l{l'})=\app{e(l)}{e(l')}$, for any $l,l'\in\lJms$. For
$\sigma$-steps the situation is different: one $\sigma$-step in
$\lJms$ gives rise to one $\sigma$-step in $\lJmse$ but also,
possibly, to $\pi$ and $\epsilon$ steps. We consider below the base
case of $\sigma$-reduction. The following two observations are
needed:
\begin{center}
\begin{tabular}{ll}
(1)& $(y)e(t)[] \rightarrow^*_{\pi} e((y)t)$, for any $t\in\lJms$
and
any variable $y$;\\
(2)& $[e(t)/x]e(u)\rightarrow^*_{\pi}e([t/x]u)$, for any
$t,u\in\lJms$ and any variable $x$.
\end{tabular}
\end{center}
In the first observation, one can say more specifically that no
$\pi$-step is required if $t$ is a value and otherwise, if $t$ is a
command, exactly one $\pi$-step of the form $\ctt{c} []\rightarrow
c$ is needed (with $c$ a command). The second observation uses the
first and is proved simultaneously with its analogue for co-terms.

Let us then consider the case where we have the reduction
$t(x)v\rightarrow [t/x]v$ in $\lJms$ . We concentrate on the
sub-case $v=V$. (The other sub-case, where $v=t_0l_0$, is similar.)

\[
\begin{array}{rcll}
e(t(x)V)&=&\ctt{e(t)(x)e(V)[]}\\
&\rightarrow_{\sigma}&\ctt{[e(t)/x](e(V)[])}\\
&=&\ctt{[e(t)/x]e(V)\,[]}\\
&\rightarrow_{\epsilon}&[e(t)/x]e(V)\\
&\rightarrow^*_{\pi}&e([t/x]V)& \textrm{(Observation (2) above)}\\
\end{array}
\]
Now we need to prove: $e((y)t(x)v)\rightarrow^+ e((y)[t/x]v)$. We
consider the possible forms of $V$. Sub-sub-case $V=x$.
\[
\begin{array}{rcll}
e((y)t(x)x)&=&(y)e(t)(x)x[]\\
&\rightarrow_{\sigma}&(y)e(t)[]\\
&\rightarrow^*_{\pi}&e((y)t)& \textrm{(Observation (1) above)}\\
&=&e((y)[t/x]x)
\end{array}
\]
Sub-sub-case $V=z$, with $z$ a variable distinct of $x$:
\[
\begin{array}{rcll}
e((y)t(x)z)&=&(y)e(t)(x)z[]\\
&\rightarrow_{\sigma}&(y)z[]\\
&=&e((y)[t/x]z)&\\
\end{array}
\]
Sub-sub-case $V=\lambda z.u$:
\[
\begin{array}{rcll}
e((y)t(x)\lambda z.u)&=&(y)e(t)(x)e(\lambda z.u) []\\
&\rightarrow_{\sigma}&(y)[e(t)/x]e(\lambda z.u) []\\
&\rightarrow^*_{\pi}&(y)e([t/x](\lambda z.u)) []&
\textrm{(Observation (2) above)}\\
&=&e((y)[t/x](\lambda z.u))& \textrm{($[t/x](\lambda z.u)$ is a
value)}
\end{array}
\]
\end{proof}

%-------------------------------------------------

%++++++++++++++++++++++++++++++++++++++++++++++++++++++++++++++++
\subsection{Confluence}\label{subsec:confluence}

For many purposes, it should suffice to have local confluence, which
we do have for all the systems of this article, since in all of them,
the critical pairs are joinable. Hence, thanks to Newman's lemma, all
systems are confluent on typable terms since they are strongly
normalizing, as shown in the later sections. We also believe that the
usual methods that show the diamond property for properly defined
notions of parallel reduction would yield confluence of all our
systems. The aim of this section is to give indirect proofs for the
systems of the spectrum, by inheriting confluence that is already
known.

Firstly, we argue about confluence of $\lJ$ and $\lJm$. Secondly, we
define and study a mapping from  $\lJms$ to $\lJm$. Thirdly, we
apply the ``interpretation method'' to obtain confluence also of
$\lJms$. Finally, we do the same for $\lJms$ and $\lJmse$ in order
to infer confluence of $\lJmse$.

Confluence of $\lJ$ can be obtained from confluence of the original
system $\lJ$ in \cite{JoachimskiMatthesRTA2000} where $\pi$ is
\emph{lazy}. Below we call $\lpi$ the original lazy version of $\pi$,
which reduces $t(u,x.v)S$ only to $t(u,x.vS)$ (for $v$ a value $V$,
there is no difference between $\pi$ and $\lpi$). Confluence for $\rightarrow_{\beta\pi}$ is
obtained from confluence of $\to_{\beta\lpi}$ in the same way as in
\cite{jesLPintoTYPES03} confluence of $\to_{\beta\pi'}$ is obtained
from $\to_{\beta\lpi}$, where $\pi'$ is yet another variant of $\pi$.

%----------------------------------------------------------
\begin{thm}\label{theo:confluence-lambdaJ}
$\rightarrow_{\beta\pi}$ in $\lJ$ is confluent.
\end{thm}
%----------------------------------------------------------
\begin{proof}
Assume $t\rightarrow_{\beta\pi}^* t_1$ and
$t\rightarrow_{\beta\pi}^* t_2$. Then, also
$t\rightarrow_{\beta\lpi}^* t_1$ and $t\rightarrow_{\beta\lpi}^*
t_2$, and by confluence of $\rightarrow_{\beta\lpi}$ there exists
$t_3$ such that $t_1\rightarrow_{\beta\lpi}^* t_3$ and
$t_2\rightarrow_{\beta\lpi}^* t_3$. The facts
\begin{enumerate}
\item $t'\rightarrow_{\pi}^*\pi(t')$, for all $t'$ in $\lJ$,
\item$t'\rightarrow_{\beta\lpi}^* t''$ implies $\pi(t')\rightarrow_{\beta\pi}^*\pi(t'')$, for all $t',t''$ in $\lJ$,
\end{enumerate}
where notation $\pi(t')$ represents the $\pi$ normal form of term
$t'$ (definable by recursion on $t'$, using a very eager form of generalised application \cite{JoachimskiMatthesRTA2000}), allow to conclude that $t_1,t_2$ both $\beta\pi$-reduce to
$\pi(t_3)$.
\end{proof}

What has been said above for $\lJ$ can be recast for $\lJm$, and
confluence of $\rightarrow_{\beta\pi\mu}$ obtained from confluence
of $\rightarrow_{\beta\lpi\mu}$ \cite{jesLPintoTYPES03}. In
$\lJm$, the lazy $\pi$ rule reads $tRS\to t(\lapp{R}{S})$, where
$\lapp{(u,l)}{S}=(u,\lapp{l}{S})$, and
$\lapp{(u::l)}{S}=u::(\lapp{l}{S})$ and $\lapp{((x)t)}{S}=(x)tS$.

%----------------------------------------------------------
\begin{thm}\label{theo:confluence-lambdaJm}
$\rightarrow_{\beta\pi\mu}$ in $\lJm$ is confluent.
\end{thm}
%----------------------------------------------------------
\proof
The proof above holds if $\beta$ is replaced by $\beta\mu$. In
particular, we have
\begin{enumerate}
\item $t'\rightarrow_{\pi}^*\pi(t')$, for all $t'$ in $\lJm$,
\item$t'\rightarrow_{\beta\lpi\mu}^* t''$ implies $\pi(t')\rightarrow_{\beta\pi\mu}^*\pi(t'')$, for all $t',t''$ in
$\lJm$.\qed
\end{enumerate}

Now consider confluence of $\lJms$. In this case, we cannot rely
on a previous result of confluence for some variant of the system.
Instead, we will lift the confluence result of
\cite{jesLPintoTYPES03} to $\lJms$.
First, we define a mapping $\sm{(\_)}:\lJms\to\lJm$ in Figure~\ref{fig:def-dagger}.
\begin{figure}\caption{Translation of $\lJms$ into $\lJm$}\label{fig:def-dagger}
\begin{eqnarray*}
\sm{x}&=&x\\
\sm{(\lambda x.t)}&=&\lambda x.\sm{t}\\
\sm{(t(x)v)}&=&[\sm{t}/x]\sm{v}\\
\sm{(t(u::l))}&=&\sm{t}(\sm{u},\sm{l})\\
\sm{((x)v)}&=&(x)\sm{v}\\
\sm{(u::l)}&=&\sm{u}::\sm{l}\\
\end{eqnarray*}
\end{figure}

%--------------------------------------------------------
\begin{prop}\label{prop:properties-mapping-Jms-to-Jm}\quad
\begin{enumerate}
\item For all $t\in\lJms$, $t\to_{\sigma}^*s(\sm{t})$.
\item
If $t\to u$ in $\lJms$, then $\sm{t}\to_{\beta\lpi\mu}^*\sm{u}$ in
$\lJm$.
\end{enumerate}
\end{prop}
%--------------------------------------------------------
\begin{proof}
1. The claim is proved together with the similar claim for
$l\in\lJms$ by simultaneous induction on $t$ and $l$.

2. The claim is proved together with the similar claim for
$l\to l'$ in $\lJms$. The proof is by simultaneous induction on
$t\to u$ and $l\to l'$. The proof uses the following facts:

(i) $(\lambda
x.\sm{t})(\sm{u},\sm{l})\to_{\beta}[\sm{u}/x]\sm{(tl)}$.

(ii)
$\sm{(tl_1)}(\sm{u},\sm{l_2})\to_{\lpi}^*\sm{(t(\app{l_1}{(u::l_2)}))}$
and
$\lapp{\sm{l_1}}{(\sm{u},\sm{l_2})}\to_{\lpi}^*\sm{(\app{l_1}{u::l_2})}$.

(iii) $[\sm{t}/x]\sm{v}=\sm{([t/x]v)}$.

(iv) $(x)\sm{(xl)}\to_{\mu}^=\sm{l}$, if $x\notin
l$.\footnote{$\to_R^=$ denotes the reflexive closure of $\to_R$.}

(i) and (iv) are proved by case analysis of $l$. (ii) is proved by
induction on $l_1$. (iii) is proved together with
$[\sm{t}/x]\sm{l}=\sm{([t/x]l)}$ by simultaneous induction on $v$
and $l$.\footnote{In $\lJm$ one has to use $\lpi$ and not $\pi$
for statement (2) to hold. Consider the $\lJms$-terms
$v_0=t_0(u_0::(x)(t_1(z)z))(u::k)$ and
$v_1=t_0(u_0::(x)t_1(z)z(u::k))$. Then $v_0\to_{\pi}v_1$ but
$\sm{v_0}\to_{\pi}\sm{v_1}$ fails.}
\end{proof}

%----------------------------------------------------------
\begin{thm}\label{theo:confluence-lambdaJms}
$\rightarrow_{\beta\pi\sigma\mu}$ in $\lJms$ is confluent.
\end{thm}
%----------------------------------------------------------
\begin{proof}
Suppose $t\to_{\beta\pi\sigma\mu}^*t_i$, $i=1,2$, in $\lJms$. By
part 2 of Proposition \ref{prop:properties-mapping-Jms-to-Jm},
$\sm{t}\to_{\beta\lpi\mu}^*\sm{t_i}$ in $\lJm$. By confluence
\cite{jesLPintoTYPES03}, there is $u\in\lJm$ such that
$\sm{t_i}\to_{\beta\lpi\mu}^*u$. By property 2 in the proof of
Theorem \ref{theo:confluence-lambdaJm}, we get
$\sm{t_i}\to_{\beta\pi\mu}^*\pi(u)$. By the properties of mapping
$s:\lJm\to\lJms$, we get
$s(\sm{t_i})\to_{\beta\pi\sigma\mu}^*s(\pi(u))$. We close the
diagram in $\lJms$ because $t_i\to_{\sigma}^*s(\sm{t_i})$.
\end{proof}

Finally we consider confluence of $\lJmse$. We will lift
confluence of $\lJms$. First, we define a mapping
$\es{(\_)}:\lJmse\to\lJms$ in Figure~\ref{fig:def-circ}
\begin{figure}\caption{Embedding of $\lJmse$ into $\lJms$}\label{fig:def-circ}
\begin{eqnarray*}
\es{x}&=&x\\
\es{(\lambda x.t)}&=&\lambda x.\es{t}\\
\es{\ctt{tl}}&=&\es{t}\es{l}\\
\es{[]}&=&(x)x\\
\es{((x)tl)}&=&(x)\es{t}\es{l}\\
\es{(u::l)}&=&\es{u}::\es{l}\\
\end{eqnarray*}
\end{figure}
whose intuitive idea
is that, in some
sense, $\lJmse$ is a subsystem of $\lJms$ -- precisely the
subsystem where selection is restricted to the cases $(x)x$ and
$(x)tl$.

%--------------------------------------------------------
\begin{prop}\label{prop:properties-mapping-Jmse-to-Jms}\quad
\begin{enumerate}
\item For all $t\in\lJmse$, $e(\es{t})\to_{\mu}^*t$.
\item
If $t\to u$ in $\lJmse$, then $\es{t}\to^+\es{u}$ in $\lJms$.
\end{enumerate}
\end{prop}
%--------------------------------------------------------
\begin{proof}
Claim 1 is proved together with the similar claim
$e(\es{l})\to_{\mu}^*l$, by simultaneous induction on $t$ and $l$.

Claim 2 for $\mu$ and $\epsilon$ is a direct verification. Since there are no
commands in $\lJms$, one would have to study always two versions of $\beta$,
$\pi$ and $\sigma$: once inside braces $\{\}$, once bound by $(y)$. However,
since all three rules have the form $t_1l_1\to t_2l_2$, it suffices to verify $\es
{t_1}\es{l_1}\to^+\es {t_2}\es{l_2}$ for them. For $\sigma$, we also need the
facts $\es{([t/x]u)}=[\es t/x]\es u$ and $\es{([t/x]l)}=[\es t/x]\es l$, and for the
non-nil case of $\pi$, the fact $\app{\es l}{\es{(u_1::l_1)}}\to_\mu\es{(\app
  l{u_1::l_1})}$ is proved by induction on $l$.
\end{proof}

The first statement of the previous proposition is an obstacle to an
immediate application of the ``interpretation method'', because the
$\mu$-reduction goes in the wrong direction. We overcome this by
observing that, as a consequence of $e(\es{t})\to_{\mu}^*t$, we have
$t\to_{\mu}^*\mu(e(\es{t}))$. (Here $\mu$ is the function that
assigns the $\mu$-normal form of an expression. Clearly, reduction
rule $\mu$ alone is terminating and locally confluent, hence
confluent.) So, in the proof of confluence (Theorem
\ref{thm:confluence-lambdaJmse} below) there will be an extra step
relying on the properties of mapping $\mu$, which is explicitly
given in Figure~\ref{fig:def-mu}.
\begin{figure}\caption{Description of $\mu$-normalisation function in $\lJmse$}\label{fig:def-mu}
\begin{eqnarray*}
\mu{x}&=&x\\
\mu{(\lambda x.t)}&=&\lambda x.\mu{t}\\
\mu{\ctt{tl}}&=&\ctt{\mu{t}\,\mu{l}}\\
\mu{[]}&=&[]\\
\mu{((x)tl)}&=&\mu{l}\,\,\textrm{(if $t=x$ and $x\notin l$)}\\
\mu{((x)tl)}&=&(x)\mu{t}\mu{l}\,\,\textrm{(otherwise)}\\
\mu{(u::l)}&=&\mu{u}::\mu{l}\\
\end{eqnarray*}
\end{figure}

%--------------------------------------------------------
\begin{prop}\label{prop:properties-mapping-mu}\quad
In $\lJmse$, if
$t\to u$, then $\mu{t}\to^*\mu{u}$.
\end{prop}
%--------------------------------------------------------
\begin{proof}
The claim is proved together with the similar claim for $l\to l'$
by simultaneous induction on $t\to u$ and $l\to l'$. The
proof makes use of the following facts: (i)
$(x)\mu(t)\mu(l)\to_{\mu}^=\mu((x)tl)$; (ii) commutation of
mapping $\mu$ with substitution; (iii) commutation of mapping
$\mu$ with append. Fact (i) is immediate from definition. Facts
(ii) and (iii) are proved by easy inductions.
\end{proof}

%----------------------------------------------------------
\begin{thm}\label{thm:confluence-lambdaJmse}
$\rightarrow_{\beta\pi\sigma\mu\epsilon}$ in $\lJmse$ is
confluent.
\end{thm}
%----------------------------------------------------------
\begin{proof}
Suppose $t\to_{\beta\pi\sigma\mu\epsilon}^*t_i$, $i=1,2$, in
$\lJmse$. By part 2 of Proposition
\ref{prop:properties-mapping-Jmse-to-Jms},
$\es{t}\to_{\beta\pi\sigma\mu}^*\es{t_i}$ in $\lJms$. By
confluence (Theorem \ref{theo:confluence-lambdaJms}), there is
$u\in\lJms$ such that $\es{t_i}\to_{\beta\pi\sigma\mu}^*u$. By the
properties of mapping $e:\lJms\to\lJmse$, we get
$e(\es{t_i})\to_{\beta\pi\sigma\mu\epsilon}^*e(u)$. Proposition
\ref{prop:properties-mapping-mu} yields
$\mu(e(\es{t_i}))\to_{\beta\pi\sigma\mu\epsilon}^*\mu(e(u))$. We
close the diagram in $\lJmse$ because
$t_i\to_{\mu}^*\mu(e(\es{t_i}))$.
\end{proof}

Notice that we might have inferred confluence of $\lJmse$ of that of
the call-by-name $\lmmt$-calculus, presented in
Section~\ref{subsec:lambda-mu-mu-tilde}: if this calculus is
confluent, then its intuitionistic fragment is confluent as well since
it has just the same rules on a subset of terms, co-terms and commands
that is closed under reduction. Finally, its isomorphic copy $\lJmse$
would be confluent as well. However, we are not aware of a proof of
confluence of our version of call-by-name $\lmmt$-calculus: the
calculus considered in \cite{Likavec:PhD-2005} does not have the rules $\mu$
and $\mt$, has a more restrictive notion of evaluation contexts
and imposes $\sigma$-reduction immediately following applications of
$\beta$. As mentioned above, we would expect that the standard direct
proof methods would be applicable to establish confluence of all of
the systems considered in this section.

%*********************************************************************************
\section{CGPS translations}\label{sec:cgps}

In this section we define a CPS translation for $\lJmse$ into the
simply-typed $\lambda$-calculus and show how it fails to provide a
strict simulation of reduction. Next we refine the CPS translation
to a CGPS translation of $\lJmse$ and show that strict simulation
of reduction is obtained. Strong normalisation for $\lJmse$
follows. Finally, we adapt the CGPS translation to the subsystems
of $\lJmse$.

%+++++++++++++++++++++++++++++++++++++++++++++++++++++++++++
\subsection{CPS translation for $\lJmse$}\label{subsec:cps}

We assume the reader is familiar with simply-typed lambda-calculus
(we write $A\imp B$ for the function type $A\to B$ and $\to_\beta$
for the one-step reduction relation). Fix a ground type (some type
variable) $\bot$. Then, $\neg A:=A\imp\bot$, as usual in
intuitionistic logic. While our calculus is strictly intuitionistic
in nature, a double-negation translation nevertheless proves useful
for the purposes of establishing strong normalisation, as has been
shown by de Groote \cite{Groote02} for disjunction with its
commuting conversions. A type $A$ will be translated to
$\ov{A}=\neg\neg A^*,$ with the type $A^*$ defined by recursion on
$A$ (where the definition of $\ov A$ is used as an abbreviation):
\[
\begin{array}{rcl}
 X^*&=&X\\
 (A\imp B)^*&=&\neg\ov{B}\imp\neg\ov{A}\\
 \end{array}
\]

\noindent We thus obtain
\[
\begin{array}{rcl}
  \ov X&=&\neg\neg X\\
  \ov{A \imp B}&=&\neg\neg(\neg\ov{B}\imp\neg\ov{A})
\end{array}
\]

The symmetrically-looking definition of $(A\imp B)^*$ is logically
equivalent to $\ov{A}\imp\neg\neg\ov{B}$. The additional double
negation of $\ov{B}$ is needed even for weak simulation to hold.
See Subsection~\ref{subsec:simpler-cps} for a discussion of this
issue.

The translation of all syntactic elements $T$ will be presented in
Plotkin's \cite{PlotkinTLCA1975} colon notation $(T:K)$ for some
term $K$ taken from simply-typed $\lb$-calculus. A term $t$ of
$\lJmse$ will then be translated into the simply-typed $\lb$-term
$$\ov{t}=\lambda k.(t:k)$$ with a ``fresh'' variable $k$ (one that is not free in $t$). The
definition of $(T:K)$ in Figure~\ref{fig:cps} uses the definition of $\ov{t}$ as an
abbreviation (the variables $m,w$ are supposed to be ``fresh'', in the obvious sense).
\begin{figure}\caption{CPS translation of $\lJmse$}\label{fig:cps}
\[
 \begin{array}{rcl@{\qquad\qquad}rcl}
 (x:K)&=&xK&([]:K)&=&\lambda w.wK\\
 (\lambda x.t:K)&=&K(\lambda wx.w\ov{t})&(u::l:K)&=&\lambda w.w(\lambda m.m\,(l:K)\,\ov{u})\\
 (\ctt c:K)&=&(c:K)&((x)c:K)&=&\lambda x.(c:K)\\[1ex]
 &&&(t[]:K)&=&(t:K)\\
 &&&(t(u::l):K)&=&(t:\lambda m.m\,(l:K)\,\ov{u})\\
 &&&(t(x)c:K)&=&((x)c:K)\ov{t}
 \end{array}
\]
\end{figure}
The translation admits the typing rules of
Figure~\ref{fig:cpstypes}.\footnote{Regrettably, the contexts
$\Gamma'$ observed in these rules, as well as those observable
below in the rules of Fig.~\ref{fig:cgpstypes}, were missing in
\cite{ESMP:TLCA07}.}
%-----------------------------------------------------------------
\begin{figure}\caption{Admissible typing rules for CPS translation of $\lJmse$}\label{fig:cpstypes}
$$
\begin{array}{c}
\infer{\ov{\Gamma},\Gamma'\vdash (t:K):\bot}{\Gamma\vdash
t:A&\ov{\Gamma},\Gamma'\vdash K:\neg
A^*}\quad\quad\infer{\ov{\Gamma},\Gamma'\vdash
(c:K):\bot}{\seqc\Gamma cA&\ov{\Gamma},\Gamma'\vdash K:\neg A^*}\\[2ex]
\infer{\ov{\Gamma},\Gamma'\vdash
(l:K):\neg\ov{A}}{\seql{\Gamma}{l}{A}{B}&\ov{\Gamma},\Gamma'\vdash
K:\neg B^*}
\end{array}
$$
\end{figure}
%-----------------------------------------------------------------
Only the first premise in these three rules refers to $\lJmse$, the
other ones to simply-typed $\lambda$-calculus. $\ov{\Gamma}$ is
derived from $\Gamma$ by replacing every $x:C$ in $\Gamma$ by
$x:\ov{C}$. As a direct consequence (to be established during the
proof of the above typings), type soundness of the CPS translation
follows:
$$\Gamma\vdash_{\lJmse} t:A\Longrightarrow\ov{\Gamma}\vdash_{\lb}\ov{t}:\ov{A}$$

This CPS translation is also sound for reduction, in the sense that
each reduction step in $\lJmse$ translates to zero or more
$\beta$-steps in $\lambda$-calculus. Because of the collapsing of
some reductions, this result does not guarantee yet strong
normalisation of $\lJmse$.

%-----------------------------------------
\begin{prop}
\label{proposition:soundness-cps}
 If $t\rightarrow u$ in $\lJmse$,
then $\ov{t}\rightarrow^*_{\beta} \ov{u}$ in the
$\lambda$-calculus.
\end{prop}
%------------------------------------------

\begin{proof}
Simultaneously we prove
\[
T\rightarrow T' \Longrightarrow (T:K) \rightarrow^*_{\beta} (T':K)
\]
for $T,T'$ terms, co-terms or commands. More specifically, at the
base cases, the CPS translation does the following: it identifies
$\epsilon$ and $\pi$-steps, sends one $\mu$-step into zero or more
$\beta$-steps in $\lambda$-calculus and sends one $\beta$ or
$\sigma$-step into one or more $\beta$-steps in $\lambda$-calculus.
Some comments on lemmata used in this proof can be found in
the next section.
\end{proof}

%*********************************************************************************
\subsection{CGPS translation for $\lJmse$}\label{subsec:cgps}

This is the central mathematical finding of the present article. It
is very much inspired from a ``continuation and garbage passing
style'' translation for Parigot's $\lambda\mu$-calculus, proposed by
Ikeda and Nakazawa \cite{IkedaNakazawa06}. While they use garbage to
overcome the problems of earlier CPS translations that did not carry
$\beta$-steps to at least one $\beta$-step if they were under a
vacuous $\mu$-binding, as reported in \cite{NakazawaTatsuta03}, we
ensure strict simulation of $\epsilon$, $\pi$ and $\mu$. Therefore,
we can avoid the separate proof of strong normalisation of
permutation steps alone that is used in addition to the CPS in
\cite{Groote02} (there in order to treat disjunction and not for
sequent calculi as we do).

Our CGPS translation passes ``garbage'', in addition to
continuations. We mean by ``garbage'' $\lambda$-terms, denoted
$G$, that are carried around for their operational properties, not
for denotational purposes. They inhabit a type $\top$, of which we
only require that there is a term $\mathsf{s}:\top\imp\top$ such
that $\mathsf{s}\, G\to^+_\beta G$. This can of course be realized
by any type, with $\mathsf{s}:=\lambda x.x$, but it is useful, in
view of a comparison with \cite{IkedaNakazawa06}, to have in mind
another realization, namely $\top:=\perp\imp\perp$ and
$\mathsf{s}:=\lambda x.[x;\lambda z.z]$. Here we are using the
abbreviation $[t;u]:=(\lb x.t)u$ for some $x\notin t$. Then,
$[t;u]\to_\beta t$, and $\Gamma\vdash t:A$ and $\Gamma\vdash u:B$
together imply $\Gamma\vdash[t;u]:A$ (as a derived typing rule of
simply-typed $\lb$-calculus). This is a form of ``deliberate
garbage'' that is used in \cite{IkedaNakazawa06}. Instead of
$\mathsf{s}\,G$, we will write $\suc G$. We will also speak about
``units of garbage''. This is so because, in our translation,
garbage will always have one of the forms $g$ (a variable),
$\suc{g}$, $\suc{\suc{g}}$, etc. We say that $\suc{G}$ has one
more ``unit of garbage'' than $G$, or that, in $\suc{G}$, $G$ is
``incremented''. In the particular realization
$\top:=\perp\imp\perp$, $\suc{G}=_{\beta}[G;\lambda z.z]$; we may
regard $\lambda z.z$ (which lives in $\top$) as the ``unit'' that
is added to $G$. In \cite{IkedaNakazawa06}, garbage is built by
``adding'' a continuation $K$ to $G$, as in $[G;K]$.

The only change w.\,r.\,t. the type translation in CPS is that, now,
 $$\ov{A}=\top\imp\neg\neg A^*$$
is used throughout, hence, again, $X^*=X$ and $(A\imp
B)^*=\neg\ov{B}\imp\neg\ov{A}$.

We define the simply-typed $\lb$-term $(T:G,K)$ for every syntactic
construct $T$ of $\lJmse$ and simply-typed $\lb$-terms $G$ and $K$.
Then, the translation of term $t$ is defined to be
$$\ov{t}=\lambda gk.(t:g,k)$$
with ``new'' variables $g,k$, that is again used as an
abbreviation inside the recursive definition of $(T:G,K)$ in
Figure~\ref{fig:cgps} (the variables $m,w$ are again
``fresh'').\footnote{\label{foot:new-CGPS}There is a slight, but
important, difference between the \emph{definition} of the CGPS
translation presented here and that presented in \cite{ESMP:TLCA07}. In
\cite{ESMP:TLCA07}, several clauses in the definition of $(l:G,K)$
or $(c:G,K)$ contained garbage increment $\suc{G}$, whereas in the
present definition those increments are, so to speak, concentrated
in the clause for $(\ctt{c}:G,K)$. The importance of this
re-definition is that it makes the purpose of those increments
more perspicuous - see the discussion around the simulation
theorem below. For the sake of a precise connection between the
two definitions, let us write the translation of
\cite{ESMP:TLCA07} as $[T:G,K]$ and $\ov{\ov{t}}$. Then, by an
easy induction, one obtains $(t:G,K)=[t:G,K]$,
$(l:\suc{G},K)=[l:G,K]$, and $(c:\suc{G},K)=[c:G,K]$. Hence
$\ov{t}=\ov{\ov{t}}$.}

%------------------------------------------------------------
\begin{figure}\caption{CGPS translation of $\lJmse$}\label{fig:cgps}
 $$
 \begin{array}{rcl}
(x:G,K)&=&x\,\suc{G}K\\
(\lambda x.t:G,K)&=&[K(\lambda wx.w\ov{t});G]\\
(\ctt c:G,K)&=&(c:\suc{G},K)\\[1ex]
([]:G,K)&=&\lambda w.w\,GK\\
(u::l:G,K)&=&\lambda w.w\,G(\lambda m.m\,(l:G,K)\,\ov{u})\\
((x)c:G,K)&=&\lambda x.(c:G,K)\\[1ex]
 (t[]:G,K)&=&(t:G,K)\\
 (t(u::l):G,K)&=&(t:G,\lambda m.m\,(l:G,K)\,\ov{u})\\
 (t(x)c:G,K)&=&((x)c:G,K)\ov{t}
 \end{array}
 $$
\end{figure}
%------------------------------------------------------------
If one removes the garbage argument, one precisely obtains the CPS
translation. The translation admits the typing rules of
Figure~\ref{fig:cgpstypes}.

%------------------------------------------------------------
\begin{figure}\caption{Admissible typing rules for CGPS translation of $\lJmse$}\label{fig:cgpstypes}
$$
\begin{array}{c@{\quad}c}
\infer{\ov{\Gamma},\Gamma'\vdash (t:G,K):\bot}{\Gamma\vdash
t:A&\ov{\Gamma},\Gamma'\vdash G:\top&\ov{\Gamma},\Gamma'\vdash
K:\neg A^*}& \infer{\ov{\Gamma},\Gamma'\vdash
(l:G,K):\neg\ov{A}}{\Gamma|l:A\vdash B&\ov{\Gamma},\Gamma'\vdash
G:\top&\ov{\Gamma},\Gamma'\vdash K:\neg
B^*}\\[2ex]
\multicolumn{2}{c}{\infer{\ov{\Gamma},\Gamma'\vdash
(c:G,K):\bot}{\seqc{\Gamma}{c}{A}&\ov{\Gamma},\Gamma'\vdash
G:\top&\ov{\Gamma},\Gamma'\vdash K:\neg A^*}}
\end{array}
$$
\end{figure}
%------------------------------------------------------------
For $\ov{\Gamma}$ see the previous section. Therefore (and to be proven simultaneously), the CGPS translation satisfies type soundness, i.\,e.,
$\Gamma\vdash t:A$ implies $\ov{\Gamma}\vdash\ov{t}:\ov{A}$.

%------------------------------------------------
\begin{lem}\label{lemma:simulation}In $\lJmse$ the following
holds:
\begin{enumerate}
\item $[\ov{t}/x](T:G,K)\rightarrow^*_{\beta}
  ([t/x]T:[\ov{t}/x]G,[\ov{t}/x]K) $ for $T$ any $u$, $l$ or $c$, and, in particular,
  $[\ov{t}/x]\ov{u}\to_{\beta}^*\ov{[t/x]u}$.
\label{simsubst}
\item $[t/x](T:G,K)=(T:[t/x]G,[t/x]K) $ for $T$
any $u$, $l$ or $c$ such that $x\notin T$.\label{simsubst0}
\item $G$ and $K$ are subterms of $(T:G,K)$ for $T$
any
$u$, $l$ or $c$.\label{simcontains}
\item $(t:\suc{G},K)\to_{\beta}^+(t:G,K)$.\label{simgarbagereduction}
\item $(l:G,K)\ov{t}\rightarrow^*_{\beta}
(tl:G,K)$\label{simapply}
\item $\lambda x. (xl:G,K)\rightarrow^+_{\beta} (l:G,K)$ if $x\notin l,G,K$.
\label{simmu}
\item\label{simappend}
\begin{enumerate}
\item $(tl:\suc{G},\lambda m.m (l':G,K)\ov{u})
\rightarrow^+_{\beta}(t\,(\app l{(u::l')}):G,K)$
\item $(l:\suc{G},\lambda m.m (l':G,K)\ov{u})\rightarrow^+_{\beta}(\app{l}{(u::l')}:G,K)$
\end{enumerate}

\end{enumerate}
\end{lem}
%-----------------------------------------------
\begin{proof} \ref{simsubst}./\ref{simsubst0}./\ref{simcontains}. Each one by simultaneous induction on
terms, co-terms and commands. Notice that the second statement has
to be proven simultaneously, but that it follows immediately from
the particular case $T=u$ of the first statement.

\ref{simgarbagereduction}.

$$
\begin{array}{rcll}
(t:\suc{G},K) & = & [\suc{G}/g](t:g,K)&\textrm{(by \ref{simsubst0}.)}\\
&\to_{\beta}^+& [G/g](t:g,K)&(*)\\
&=&(t:G,K)&\textrm{(by \ref{simsubst0}.)}\\
\end{array}
$$
\noindent where $(*)$ is justified by the fact that $g$ occurs in
$(t:g,K)$, as guaranteed by \ref{simcontains}.

\ref{simapply}. Straightforward case analysis on $l$.

\ref{simmu}. Case analysis on $l$.

Case $l=[]$.

$$
\begin{array}{rcll}
\lb x.(x[]:G,K)&=&\lb x.(x:G,K)&\\
&=&\lb x.x\,\suc{G}K&\\
&\to_{\beta}^+&\lb x.xGK&\\
&=&([]:G,K)&\textrm{(as $x\notin G,K$)}
\end{array}
$$

Case $l=u::l'$.

$$
\begin{array}{rcll}
\lb x.(x(u::l'):G,K)&=&\lb x.(x:G,\lb m.m(l':G,K)\ov{u})&\\
&=&\lb x.x\,\suc{G}(\lb m.m(l':G,K)\ov{u})&\\
&\to_{\beta}^+&\lb x.xG(\lb m.m(l':G,K)\ov{u})&\\
&=&(u::l':G,K)&\textrm{(as $x\notin u,l',G,K$)}
\end{array}
$$

Case $l=(y)c$.

$$
\begin{array}{rcll}
\lb x.(x(y)c:G,K)&=&\lb x.(\lb y.(c:G,K))\ov{x}&\\
&\to_{\beta}&\lb x.[\ov{x}/y](c:G,K)&\\
&=&\lb y.[y/x][\ov{x}/y](c:G,K)&\\
&=&\lb y.[[y/x]\ov{x}/y](c:G,K)&\textrm{(as $x\notin c,G,K$, and by \ref{simsubst0}.)}\\
&=&\lb y.[\ov{y}/y](c:G,K)&\\
&\to_{\beta}^*&\lb y.([y/y]c:G,K)&\textrm{(by \ref{simsubst}.)}\\
&=&((y)c:G,K)&
\end{array}
$$

\ref{simappend}. (a) and (b) are proved simultaneously by
induction on $l$.

Case $l=[]$.

$$
\begin{array}{rcll}
(t[]:\suc{G},\lambda m.m (l':G,K)\ov{u})
&=&(t:\suc{G},\lambda m.m (l':G,K)\ov{u})&\\
&\to_{\beta}^+&(t:G,\lambda m.m (l':G,K)\ov{u})&\textrm{(by \ref{simgarbagereduction}.)}\\
&=&(t\,(\app {[]}{(u::l')}):G,K)&
\end{array}
$$

$$
\begin{array}{rcll}
([]:\suc{G},\lambda m.m (l':G,K)\ov{u})
&=&\lb w.w\,\suc{G}(\lambda m.m (l':G,K)\ov{u})&\\
&\to_{\beta}^+&\lb w.wG(\lambda m.m (l':G,K)\ov{u})&\\
&=&(\app{[]}{(u::l')}:G,K)&
\end{array}
$$

Case $l=u_0::l_0$.

$$
\begin{array}{rcll}
&&(t(u_0::l_0):\suc{G},\lambda m.m (l':G,K)\ov{u})&\\
&=&(t:\suc{G},\lb n.n(l_0:\suc{G},\lb m.m(l':G,K)\ov{u})\ov{u_0})&\\
&\to_{\beta}^+&(t:\suc{G},\lb n.n(\app{l_0}{(u::l')}:G,K)\ov{u_0})&\textrm{(by IH (b))}\\
&\to_{\beta}^+&(t:G,\lb n.n(\app{l_0}{(u::l')}:G,K)\ov{u_0})&\textrm{(by \ref{simgarbagereduction}.)}\\
&=&(t\,(\app {(u_0::l_0)}{(u::l')}):G,K)&
\end{array}
$$

$$
\begin{array}{rcll}
&&(u_0::l_0:\suc{G},\lambda m.m (l':G,K)\ov{u})&\\
&=&\lb w.w\,\suc{G}(\lb n.n(l_0:\suc{G},\lb m.m(l':G,K)\ov{u})\ov{u_0})&\\
&\to_{\beta}^+&\lb w.w\,\suc{G}(\lb n.n(\app{l_0}{(u::l')}:G,K)\ov{u_0})&\textrm{(by IH (b))}\\
&\to_{\beta}^+&\lb w.wG(\lb n.n(\app{l_0}{(u::l')}:G,K)\ov{u_0})&\\
&=&(\app{(u_0::l_0)}{(u::l')}:G,K)&
\end{array}
$$

Case $l=(x)v_0l_0$. Part (b) follows from the induction hypothesis
(a) for $l_0$, and part (a) is an immediate consequence of (b).
\end{proof}

%---------------------------------------------
\begin{thm}[Simulation]
\label{theorem:simulation-cgps} If $t\rightarrow u$ in $\lJmse$,
then $\ov{t}\rightarrow^+_{\beta} \ov{u}$ in the
$\lambda$-calculus.
\end{thm}
%---------------------------------------------

\begin{proof}
Simultaneously we prove:
$
T\rightarrow T' \Longrightarrow (T:G,K) \rightarrow^+_{\beta} (T':G,K)
$
for $T,T'$ terms, co-terms or commands.  We illustrate the cases
of the base rules.

Case $\beta$: $(\lambda x.t)(u::l)\rightarrow u(x)tl$.
\[
\begin{array}{rcll}
((\lambda x.t)(u::l):G,K)&=&(\lambda x.t:G,\lambda m. m
(l:G,K)\ov{u})&\\
&=&[(\lambda m. m (l:G,K)\ov{u})(\lambda
wx.w\ov{t});G]&\\
&\rightarrow^3_{\beta}&(\lambda x.
(l:G,K)\ov{t})\ov{u}\\
&\rightarrow^*_{\beta}&(\lambda x.
(tl:G,K))\ov{u}&\textrm{ (Lemma \ref{lemma:simulation}.\ref{simapply})}\\
&=&(u(x)tl:G,K)
\end{array}
\]

Case $\pi$: $\ctt{tl}E \rightarrow t\,(\app lE)$. Sub-case $E=[]$.
$$
\begin{array}{rcll}
(\ctt{tl}[]:G,K)&=&(\ctt{tl}:G,K)&\\
&=&(tl:\suc{G},K)&\\
&\rightarrow_{\beta}^+&(tl:G,K)&\textrm{(Lemma \ref{lemma:simulation}.\ref{simgarbagereduction})}\\
&=&(t\,(\app l{[]}):G,K).
\end{array}
$$

Sub-case $E=u::l'$.
$$
\begin{array}{rcll}
(\ctt{tl}(u::l'):G,K)&=&(\ctt{tl}:G,\lambda m.m (l':G,K)\ov{u})&\\
&=&(tl:\suc{G},\lambda m.m (l':G,K)\ov{u})&\\
&\rightarrow^+_{\beta}&(t\,(\app l{(u::l')}):G,K)&\textrm{(Lemma \ref{lemma:simulation}.\ref{simappend})}\\
\end{array}
$$

Case $\sigma$: $t(x)c\rightarrow [t/x]c$.
\[
\begin{array}{rcll}
(t(x)c:G,K)&=&(\lambda x.(c:G,K))\ov{t}&\\
&\rightarrow_{\beta}&
[\ov{t}/x](c:G,K)&\\
&\rightarrow^*_{\beta}&([t/x]c:G,K)&\textrm{(Lemma
\ref{lemma:simulation}.\ref{simsubst})}
\end{array}
\]

Case $\mu$: $(x)xl \rightarrow l$, if $x\notin l$.
\[((x)xl:G,K)=\lambda x. (xl:G,K)\rightarrow^+_{\beta} (l:G,K)\quad\textrm{(Lemma \ref{lemma:simulation}.\ref{simmu})}\]

Case $\epsilon$: $\quad\ctt{t[]}\rightarrow t$.
\[
\begin{array}{rcll}
(\ctt{t[]}:G,K)&=&(t[]:\suc{G},K)&\\
&=&(t:\suc{G},K)&\\
&\rightarrow_{\beta}^+&(t:G,K)&\textrm{(Lemma
\ref{lemma:simulation}.\ref{simgarbagereduction})}
\end{array}
\]

The cases corresponding to the closure rule $t\rightarrow
t'\Longrightarrow tl\rightarrow t'l$ (resp.  $l\rightarrow
l'\Longrightarrow tl\rightarrow tl'$) can be proved by case analysis
on $l$ (resp. $l\rightarrow l'$). The cases corresponding to the
other closure rules follow by routine induction.
\end{proof}
\begin{rem}\rm
  Unlike the failed strict simulation by CPS reported in
  \cite{NakazawaTatsuta03} that only occurred with the closure rules,
  the need for garbage in our translation is already clearly visible
  in the subcase $E=[]$ for $\pi$ and the case $\epsilon$. But the
  garbage is also effective for the closure rules, where the most
  delicate rule is the translation of $t(u::l)$ that mentions $l$ and
  $u$ only in the continuation argument $K$ to $t$'s translation.
  Lemma \ref{lemma:simulation}.\ref{simcontains} is responsible for
  propagation of strict simulation. The structure of our garbage --
  essentially just ``units of garbage'' -- can thus be easier than in
  the CGPS in \cite{IkedaNakazawa06} for $\lb\mu$-calculus since
  there, $K$ cannot be guaranteed to be a subterm of $(T:G,K)$, again
  because of the problem with void $\mu$-abstractions. The solution of
  \cite {IkedaNakazawa06} for the most delicate case of application is
  to copy the $K$ argument into the garbage. We do not need this in
  our intuitionistic calculi. However, since we need garbage for some
  base cases, we also had to make sure that reductions in garbage
  arguments are not lost during propagation through the closure rules.
\end{rem}

Let us compare the CPS and CGPS translations in order to
understand how garbage-passing ensures strict simulation. The analogue
to Lemma \ref{lemma:simulation} for the CPS translation is
obtained by erasing garbage throughout, and replacing
$\to_{\beta}^+$ by equality in items \ref{simgarbagereduction}
and \ref{simappend}, and by $\to_{\beta}^*$ in item \ref{simmu}.
So, the properties of the CGPS translation are as ``good'' as
those of the CPS translation, and at least a weak simulation
could be expected.

An inspection of the proofs of Lemma \ref{lemma:simulation} and
Theorem \ref{theorem:simulation-cgps} shows that the CGPS
translation generates reduction sequences which, so to speak,
differ from those generated by CPS translation by the insertion of
sequences of the form $\suc{G}\to_{\beta}^+G$. The point is that
the CGPS translation does such insertion at all points where the
CPS does an undesired identification (although it also does at
other points where such insertion is unnecessary).

The ultimate cause for the existence of such dynamic \emph{garbage
decrement steps} is the static garbage increment contained in the
clauses defining $(x:G,K)$ and $(\ctt{c}:G,K)$. Moreover, it can
be argued that the clause for $(x:G,K)$ is responsible for strict
simulation of $\mu$-steps, whereas the clause for $(\ctt{c}:G,K)$
is the cause for strict simulation of $\pi$- or $\epsilon$-steps.

The key for strict simulation of $\mu$-steps is Lemma
\ref{lemma:simulation}.\ref{simmu}. An inspection of the proof
shows that garbage plays no role in the case $l=(y)c$ (which
already generated reduction steps through the CPS translation),
and that, had $(x:G,K)$ been defined as $xGK$, the same
identifications obtained before with the CPS translation would
have arisen again in the cases $l=[]$ and $l=u::l'$. The
definition of $(x:G,K)$ causes many garbage decrement steps, which
are useless most of the time (typically adding to the
administrative steps, already generated in the case of the CPS
translations, that mediate between $[\ov{t}/x]\ov{u}$ and
$\ov{[t/x]u}$), but not so in the particular situations described
in the cases $l=[]$ and $l=u::l'$ of Lemma
\ref{lemma:simulation}.\ref{simmu}.

The role of clause $(\ctt{c}:G,K)$ is plain for $\epsilon$ and the
case $E=[]$ of $\pi$. As to the case $E=u::l'$, it suffices to
observe that $(tl:G,\lambda m.m (l':G,K)\ov{u})=(t\,(\app
l{(u::l')}):G,K)$ (as an inspection of the proof of Lemma
\ref{lemma:simulation}.\ref{simappend} easily shows). So, again, had
$(\ctt{c}:G,K)$ been defined as $(c:G,K)$, the same identifications
of $\pi$- or $\epsilon$-steps obtained before with the CPS
translation would have arisen. The definition of $(\ctt{c}:G,K)$
means that garbage-passing does, among other things, some form of
counting braces. The braces decrement observed in $\pi$- or
$\epsilon$-steps in the source is reflected by garbage decrement
steps in the target.

%---------------------------------
\begin{cor}\label{cor:SN-for-LambdaJmse}
The typable terms of $\lJmse$ are strongly normalising.
\end{cor}
%---------------------------------

Recalling our discussion in Section~\ref{sec:lambdaJmse}, we already
could have inferred strong normalisation of $\lJmse$ from that of
$\lmmt$, which has been shown directly by Polonovski
\cite{PolonovskiFOSSACS2004} using reducibility candidates and
before by Lengrand's \cite{LengrandWRS2003} embedding into a
calculus by Urban that also has been proven strongly normalizing by
the candidate method. Our proof is just by a syntactic
transformation to simply-typed $\lb$-calculus.

%--------------------------------------------------------------------

Since each of $m$, $s$ and $e$ preserves typability and strictly
simulates reduction (Proposition~\ref{prop:simulation-by-embeddings}),
it follows from Corollary \ref{cor:SN-for-LambdaJmse} that:

%---------------------------------
\begin{cor}\label{cor:SN-for-lambdaJ(m)(s)}
The typable terms of $\lJms$, $\lJm$ and $\lJ$ are strongly
normalising.
\end{cor}
%---------------------------------

%+++++++++++++++++++++++++++++++++++++++++++++++++++++++++++
\subsection{CGPS translations for subsystems}\label{subsec:cgps-for-subsystems}

We define CGPS translations for $\lJms$, $\lJm$ and $\lJ$. The
translation of types is unchanged. In each translation, we just
show the clauses that are new.

\begin{enumerate}

\item For $\lJms$:
$$
\begin{array}{rcl}
(tl:G,K)&=&(tl;\suc{G},K)\\[1ex]
((x)V:G,K)&=&\lambda x.(V:G,K)\textrm{ ($V$ a value)}\\
((x)tl:G,K)&=&\lambda x.(tl;G,K)\\[1ex]
(t(x)v;G,K)&=&((x)v:G,K)\ov{t}\\
(t(u::l);G,K)&=&(t:G,\lambda m.m\,(l:G,K)\,\ov{u})\\
\end{array}
$$

\item For $\lJm$: there is no
auxiliary operator $(tl;G,K)$.
$$
\begin{array}{rcl}
(t(u,l):G,K)&=&(t:\suc{G},\lambda m.m\,(l:\suc{G},K)\,\ov{u})\\
((x)t(u,l):G,K)&=&\lb x.(t:G,\lambda m.m\,(l:G,K)\,\ov{u})\\
\end{array}
$$

\item Finally, for $\lJ$:
$$
\begin{array}{rcll}
(t(u,x.V):G,K)&=&(t:\suc{G},\lambda m.m\,(\lambda
x.(V:\suc{G},K))\,\ov{u})&\textrm{($V$ a value)}\\
(t(u,x.v):G,K)&=&(t:\suc{G},\lambda m.m\,(\lambda
x.(v:G,K))\,\ov{u})&\textrm{($v$ an application)}
\end{array}
$$

\end{enumerate}

\noindent In the case of $\lJms$, the distinction between
$(tl:G,K)$ and $(tl;G,K)$ is consistent with the distinction, in
$\lJmse$, between $(\ctt{c}:G,K)$ and $(c:G,K)$.\footnote{We take
the opportunity to correct a mistake in the CGPS translations for
the subsystems of $\lJmse$ given in \cite{ESMP:TLCA07}. The
mistake was that some clauses in the definition of those
translations lacked a needed case analysis. We repair the mistake
now, using in this footnote the notations $(T:G,K)$ and $\ov{t}$
with their meanings in \cite{ESMP:TLCA07}. For $\lJms$:
$((x)v:G,K)=\lambda x.(v:G',K)$, where $G'$ is either $\suc{G}$,
if $v$ is a value; or $G$, otherwise. For $\lJm$, it should be
understood that the clause for $((x)v:G,K)$ just given is
inherited. For $\lJ$: $(t(u,x.v):G,K)=(t:\suc{G},\lambda
m.m\,(\lambda x.(v:G',K))\,\ov{u})$, where $G'$ is $\suc{G}$, if
$v$ is a value; or $G$, otherwise.}

These translations  are coherent with the CGPS translation for
$\lJmse$:

%----------------------------------------------------------
\begin{prop}
\label{proposition:coherence-e} Let
$\mathcal{L}\in\{\lJms,\lJm,\lJ\}$. Let $f_{\mathcal{L}}$ be the
embedding of $\mathcal{L}$ in the immediate extension of
$\mathcal{L}$ in the spectrum of Fig.~\ref{fig:spectrum}, and let
$g_{\mathcal{L}}$ be the embedding of $\mathcal{L}$ in $\lJmse$.
Then, for all $t\in \mathcal{L}$, $\ov{t}=\ov{f_{\mathcal{L}}(t)}$.
Hence, for all $t\in \mathcal{L}$, $\ov{t}=\ov{g_{\mathcal{L}}(t)}$.
\end{prop}
%----------------------------------------------------------

\begin{proof}

For $\lJms$, let $P(t):=\forall G,K((t:G,K)=(e(t):G,K))$, for
every $t\in\lJms$. Then, one proves

\begin{tabular}{l}
(i) $P(t)$; and\\
(ii) $(l:G,K)=(e(l):G,K)$ and $\forall
t\in\lJms(P(t)\Longrightarrow (e(t)e(l):G,K)=(tl;G,K))$
\end{tabular}

\noindent  by simultaneous induction on $t$ and $l$.

For $\lJm$, on proves $(t:G,K)=(s(t):G,K)$ and
$(l:G,K)=(s(l):G,K)$ by simultaneous induction on $t$ and $l$.

For $\lJ$, one proves $(t:G,K)=(m(t):G,K)$ by induction on $t$.
\end{proof}

Therefore, since each of $m$, $s$ and $e$, as well as the CGPS
translation of $\lJmse$, preserves typability, so does each CGPS
translation of the subsystems.

%-------------------------------------------
\begin{thm}[Simulation] Let $\mathcal{L}\in\{\lJms,\lJm,\lJ\}$.
If $t\rightarrow u$ in $\mathcal{L}$, then
$\ov{t}\rightarrow^+_{\beta} \ov{u}$ in the $\lambda$-calculus.
\end{thm}

%-------------------------------------------
\begin{proof}
By Propositions \ref{prop:simulation-by-embeddings} and
\ref{proposition:coherence-e} and Theorem
\ref{theorem:simulation-cgps}.
\end{proof}

Since the various CGPS translations preserve typability, Corollary
\ref{cor:SN-for-lambdaJ(m)(s)} follows also from the previous
theorem (and strong normalisation of the simply-typed
$\lambda$-calculus).

The CGPS translations defined above for the subsystems of
$\lJmse$, being consistent with the CGPS translation of $\lJmse$,
have the advantage of inheriting the simulation theorem, and the
disadvantage of not being optimized for the particular system on
which they are defined. In fact, such translations can be optimized by
omitting garbage increments $\suc{G}$ in one or more of their
clauses. We give one example of this phenomenon.

%-------------------------------------------
\begin{thm}[Simulation] \label{theorem:simulation-lJ} Let $\ov{t}$ and $(t:G,K)$ be given for $t\in\lJ$ by:

$$
\begin{array}{rcl}
\ov{t}&=&\lb gk.(t:g,k)\\
(x:G,K)&=&xGK\\
(\lb x.t:G,K)&=&[K(\lb wx.w\ov{t});G]\\
(t(u,x.v):G,K)&=&(t:\suc{G},\lambda m.m\,(\lambda
x.(v:G,K))\,\ov{u})
\end{array}
$$

If $t\rightarrow u$ in $\lJ$, then $\ov{t}\rightarrow^+_{\beta}
\ov{u}$ in the $\lambda$-calculus.
\end{thm}
%-------------------------------------------
\begin{proof} Similar to, but simpler than that of Theorem
\ref{theorem:simulation-cgps}.
\end{proof}

%+++++++++++++++++++++++++++++++++++++++++++++++++++++++++++
\subsection{C(G)PS translations with less double negations}\label{subsec:simpler-cps}

  Our definition of $(A\imp B)^*$ produces a type logically equivalent to $\ov{A}\imp\neg\neg\ov{B}$, which
  has an extra double negation of $\ov{B}$ when compared with traditional CPS's. One may wonder what happens if
  one sets $(A\imp B)^*=\ov A \imp \ov B$. There is no problem in \emph{defining} a CPS translation
  based on that, but we would even lose weak simulation in the form of
  Proposition~\ref{proposition:soundness-cps}. Let us take the simplified
  type translation, whose new clauses are:
\[
\begin{array}{rcl}
(\lambda x.t:K)&=&K(\lb x.\ov t)\\
(u::l:K)&=&\lambda w.w(\lambda m.(l:K)(m\ov{u}))\\
(t(u::l):K)&=&(t:\lambda m.(l:K)(m\ov{u}))
\end{array}
\]
This translation obeys to the typing rules of
Figure~\ref{fig:cpstypes}, but already $\beta$ steps at the root do
not obey to Proposition~\ref{proposition:soundness-cps}:

\begin{eqnarray*}
((\lambda x.t)(u::[]):K)
&=&(\lb m.(\lb w.wK)(m\ov u))(\lb x.\ov t)\\
&\to_{\beta}^2&(\lb x.\ov t)\ov u K\\
&=_{\beta}&(\lb x.\ov t K)\ov u\\
&\to_{\beta}&(\lb x.(t:K))\ov u\\
&=&(u(x)(t[]):K)
\end{eqnarray*}

\noindent The problem is that there is no reduction step from
$(\lb x.\ov t)\ov u K$ to $(\lb x.\ov t K)\ov u$ in
$\lb$-calculus. Similar remarks apply to the subsystem $\lJms$.

The failed simulation just illustrated would have not occurred,
had the $\beta$ rule been defined with implicit substitution:
$$(\lb x.\ov t)\ov u K
\to_{\beta}([\ov u/x]\ov t)K
=[\ov u/x](\ov tK)
\to_{\beta}[\ov u/x](t:K)
\to_{\beta}^*([u/x]t:K)
=(\ctt{[u/x]t}[]:K)$$

\noindent This is consistent with another fact: weak simulation,
through the simpler CPS, is recovered as soon as one moves down in
the spectrum to $\lJm$ or $\lJ$ (systems where $\beta$-reduction
employs implicit substitution). For these systems, the combination
of garbage passing with the simpler CPS delivers strict
simulation. The theorem below exemplifies the situation with $\lJ$
(to be compared with Theorem \ref{theorem:simulation-lJ}).

%-------------------------------------------
\begin{thm}[Simulation]
For a type $A$, let $\ov{A}=\top\imp\neg\neg A^*$, $X^*=X$ and
$(A\imp B)^*=\ov{A}\imp\ov{B}$ and for $t\in\lJ$, let
$\ov{t}=\lambda gk.(t:g,k)$ and $(t:G,K)$ be defined as

$$
\begin{array}{rcl}
\ov{t}&=&\lb gk.(t:g,k)\\
(x:G,K)&=&xGK\\
(\lb x.t:G,K)&=&[K(\lb x.\ov{t});G]\\
(t(u,x.v):G,K)&=&(t:\suc{G},\lambda m.(\lambda x.(v:G,K))(m\ov{u}))
\end{array}
$$

\begin{enumerate}
\item If $\Gamma\vdash_{\lJ} t:A$ then
$\ov{\Gamma}\vdash_{\lb}\ov{t}:\ov{A}$.
\item If $t\rightarrow u$ in $\lJ$, then $\ov{t}\rightarrow^+_{\beta}
\ov{u}$ in the $\lambda$-calculus.
\end{enumerate}
\end{thm}
%-------------------------------------------
\begin{proof}
The proof of (1) is based on the fact that the first rule of
Figure~\ref{fig:cgpstypes} is still admissible. Property (2) follows
along the lines of theorems \ref{theorem:simulation-lJ} and
\ref{theorem:simulation-cgps}, requiring some properties analogous
to those in Lemma \ref{lemma:simulation}. We illustrate below the
base case for $\beta$ (problematic for $\lJmse$ and $\lJms$, as
explained above):

\[
\begin{array}{rcll}
((\lambda x.t)(u,y.v):G,K)&=&(\lambda x.t:\suc G,\lambda m. (\lambda
y.
(v:G,K))(m\ov{u}))\\
 &=&[(\lambda m. (\lambda y.
(v:G,K))(m\ov{u}))(\lambda x.\ov{t});\suc G]&\\
&\rightarrow^4_{\beta}&
[[\ov{u}/x]\ov{t}]/y](v:G,K)\\
&\rightarrow^*_{\beta}& [\ov{[u/x]t]}/y](v:G,K)&\\
&\rightarrow^*_{\beta}& ([[u/x]t]/y]v:G,K)&
\end{array}
\]
\end{proof}
Note that this translation of types, variables and $\lb$-abstractions
coincides with that of~\cite{IkedaNakazawa06}. Evidently, the case of
generalized application is new since it was not considered there. Only
here is the need for a garbage increment.

Finally, let us observe that the extra double negation in $(A\imp
B)^*$ has to be integrated as $\neg\ov B\imp\neg\ov A$, and not as
$\ov A\imp\neg\neg\ov B$. Had the latter alternative been adopted,
and again, already for CPS, we would lose weak simulation. The
CPS would then be defined by:
\[
\begin{array}{rcl}
(\lambda x.t:K)&=&K(\lb xw.w\ov t)\\
(u::l:K)&=&\lambda w.w(\lambda m.m\ov u(l:K))\\
(t(u::l):K)&=&(t:\lambda m.m\ov u(l:K))
\end{array}
\]

\noindent With these definitions, one calculates:
\begin{eqnarray*}
((\lambda x.t)(u::[]):K)
&=&(\lb m.m\ov u(\lb w.wK))(\lb xw.w\ov t)\\
&\to_{\beta}&(\lb xw.w\ov t)\ov u(\lb w.wK)\\
&=_{\beta}&(\lb x.(\lb w.w\ov t)(\lb w.wK))\ov u\\
&\to_{\beta}^3&(\lb x.(t:K))\ov u\\
&=&(u(x)(t[]):K)
\end{eqnarray*}
Again, the undirected $=_\beta$-step cannot be dispensed with by
reduction steps in $\lb$-calculus.

%*********************************************************************************
\section{Higher-Order Systems}\label{sec:higher-order-systems}

In this section, we extend the CGPS translation to, and obtain
strong normalisation for, the extensions of $\lJmse$ described in
the following table:\\

\begin{tabular}{r|r|l}
intuitionistic logic & sequent calculus & natural deduction system\\
\hline
propositional&$\lJmse$ & $\lb$\\
second-order propositional&$\lJmseS$ & $\ldfS$\\
higher-order propositional&$\lJmseO$ & $\ldfO$\\
higher-order predicate&$\HJmse$ & $\HOL$
\end{tabular}\\

\noindent Such extensions constitute several systems of
intuitionistic logic formulated as sequent calculi, and have a
corresponding natural deduction system. The latter are formulated
as domain-free type theories \cite{BartheSoerensen00}, and all but
one belong to the domain-free cube. The only exception is $\HOL$,
which is a domain-free formulation of Geuvers' treatment of
higher-order intuitionistic logic \cite{GeuversCahiers}.

Each CGPS translation goes from a sequent calculus to the
corresponding natural deduction system, where the latter is
expected to satisfy strong normalisation. This is the case for the
systems in the domain-free cube \cite{BartheSoerensen00}. As to
$\HOL$, it is well known that it is a pure type system
\cite{GeuversCahiers} which, in addition, has a functional
specification \cite{BartheSoerensen00}. Now \emph{op.cit.} shows
that in such cases, strong normalisation of the domain-full system
implies the same property for the domain-free one. Therefore, we
infer from the strong normalisation of Geuvers' system that of
$\HOL$.

The formulation of the systems of higher order (unlike those at
second order) require the introduction of an upper level of
domains of quantification and their inhabitants. In order to avoid
that these technicalities blur the simplicity by which the
properties of the CGPS extend beyond the (zero-th order)
propositional case, we decided to develop first the second-order
case with the simplest formulation, even at the price of a little
amount of redundancy.

%++++++++++++++++++++++++++++++++++++++++++++++++++++++++++++++++++++++
\subsection{Second Order}\label{subsec:snd-order}

All the results of the previous sections readily extend to the second
order which is one of the important advantages of double-negation
translations w.\,r.\,t. G\"odel's negative translation (employed for
first-order $\lb\mu$-calculus by Parigot \cite{Parigot97}).
In order to give an idea of how to proceed,
we will sketch how to equip $\lJmse$ by a second-order
universal quantifier (yielding system $\lJmseS$) and how to extend the
CGPS translation of $\lJmse$ into simply-typed $\lb$-calculus to a
translation of $\lJmseS$ into a ``domain-free'' version
$\ldfS$ \cite{BartheSoerensen00} of second-order
$\lb$-calculus a.\,k.\,a.~system~$F$ \cite{GirardF}.

\subsubsection{System $\ldfS$}
To recall, system~$F$ corresponds to
second-order propositional logic and consequently also has the
types of the form $\forall X.A$. Therefore, also on the type
level, we need to allow silent renaming of bound variables. Just
as it is done in \cite{IkedaNakazawa06}, we stay with the
Curry-style typing of our previous systems but nevertheless add
$\Lambda X.t$ and $tA$ to the term syntax for $\lb$, for universal
introduction and universal elimination, respectively. These two
constructions normally belong to the typing discipline \`a la
Church, but in addition to $\lb$, they give (a variant of) system
$\ldfS$ of \cite{BartheSoerensen00}. The new typing rules are:

\[\infer[RIntro2]{\Gamma\vdash \Lambda X.t:\forall X.A}{\Gamma\vdash t:A}\qquad\infer{\Gamma\vdash tB:[B/X]A}{\Gamma\vdash t:\forall X.A}
\]
with $[B/X]A$ denoting type substitution, where $RIntro2$ is under the
proviso that $X$ is not free in any type in $\Gamma$. The new
reduction rule is $\beta2$:
$$(\Lambda X.t)B\to [B/X]t$$
with $[B/X]T$ type substitution in term $t$. It is shown in
\cite{BartheSoerensen00} that strong normalisation of typable terms
is inherited from the same property for system~$F$, that has been
established by Tait's refinement \cite{Tait75} of Girard's weak
normalisation proof \cite{GirardF}.

\subsubsection{System $\lJmseS$}
For $\lJmseS$, we also extend the term
syntax by $\Lambda X.t$ and extend the co-term syntax by $A::l$
that count among the evaluation contexts. The cases $u::l$ and
$A::l$ can be uniformly seen as $U::l$, where $U$ now stands for a
term or type. Type substitution $[B/X]T$ for $T$ a
term/co-term/command can be defined in the obvious way. $\lJmseS$
extends $\lJmse$ also by the rule $RIntro2$ above and by
$$\infer[LIntro2]{\seql{\Gamma}{B::l}{\forall X.A}{C}}{\seql{\Gamma}{l}{[B/X]A}{C}}$$
The notion $\app l{l'}$ is redefined with $u$ replaced by $U$ (and
stays associative), and the admissible typing rules for substitution,
weakening and $\app{}{}$ carry over from $\lJmse$, as well as the
obvious typing rules for type substitution.  The only new reduction
rule is
$$(\beta2)\quad(\Lambda X.t)(B::l)\to([B/X]t)l$$
So, term substitution is dealt with in an explicit way in $\lJmseS$,
but type substitution is still left implicit. This gap would be
annoying for dependently-typed systems, see \cite{LengrandEtAl06} for
a proposal that solves this problem.

The one-step reduction relation takes into account the new syntactic
constructions, and subject reduction follows.

\subsubsection{CGPS translation}
The CGPS-translation of $\lJmse$ into
$\lb$ is now extended to a CGPS of $\lJmseS$ into $\ldfS$. Unlike the case of implication, no further double negation
w.\,r.\,t.~\cite{IkedaNakazawa06} has to be added, since our sequent calculi do not provide an explicit type substitution; we set
$$(\forall X.A)^*=\forall X.\ov{A}\enspace.$$
Evidently, $([B/X]A)^*=[B^*/X]A^*$, hence (but to be proven
simultaneously) $\ov{[B/X]A}=[B^*/X]\ov{A}$. We extend the definition
of $(T:G,K)$ for $\lJmse$, with $G,K$ terms of $\ldfS$, by
$$
\begin{array}{rcl}
  (\Lambda X.t:G,K)&=&[K(\Lambda X.\ov{t});G]\\
  (t(B::l):G,K)&=&(t:G,\lambda m.(l:G,K)(mB^*))\\
  (B::l:G,K)&=&\lb w.wG(\lambda m.(l:G,K)(mB^*))
\end{array}
$$
The clause for $\Lambda X.t$ is taken from \cite{IkedaNakazawa06}.
This extended translation obeys to the same typing as for $\lJmse$ (now always w.\,r.\,t.~$\ldfS$),
hence satisfies type soundness.

%------------------------------------------------
\begin{lem}\label{lemma:simulation-2nd-order}The CGPS translation of $\lJmseS$ into $\ldfS$ satisfies the following:
\begin{enumerate}[\em(1)]
\item \emph{$-$ (7)} as in Lemma \ref{lemma:simulation}.
\item[\em(8)] $[B^*/X](T:G,K)=([B/X]T:[B^*/X]G,[B^*/X]K)$ and
$[B^*/X]\ov{t}=\ov{[B/X]t}$.
\item[\em(9)] $[B/X](T:G,K)=(T:[B/X]G,[B/X]K)$ for $X$ not free in $T$.

\item[\em(10)] \begin{enumerate}[\em(a)]
\item
$(tl:\suc{G},\lambda m.(l':G,K)(mB^*))\to_\beta^+(t(\app l{(B::l')}):G,K)$
\item
$(l:\suc{G},\lambda m.(l':G,K)(mB^*))\to_\beta^+(\app l{(B::l')}:G,K)$
\end{enumerate}
\end{enumerate}
\end{lem}
%-----------------------------------------------

Now we prove the following theorem just as before
Theorem~\ref{theorem:simulation-cgps}.

%---------------------------------------------
\begin{thm}[Simulation]
\label{theorem:simulation-cgps2} If $t\rightarrow u$ in $\lJmseS$,
then $\ov{t}\rightarrow^+_{\beta} \ov{u}$ in $\ldfS$.
\end{thm}
%----------------------------------------------

\proof We show the new base cases.

Case $\beta 2$: $(\Lambda X.t)(B::l)\rightarrow([B/X]t)l$.
\[
\begin{array}{rcll}
&&((\Lambda X.t)(B::l):G,K)&\\
&=&(\Lambda X.t:G,\lambda m. (l:G,K)(m B^*))&\\
&=&[(\lambda m. (l:G,K)(m B^*))(\Lambda X.\ov{t});G]&\\
&\rightarrow^3_{\beta}&(l:G,K)[B^*/X]\ov{t}&\\
&=&(l:G,K)\ov{[B/X]t}&(\textrm{Lemma \ref{lemma:simulation-2nd-order}.8.})\\
&\rightarrow^*_{\beta}&(([B/X]t)l:G,K)&\textrm{ (Lemma
\ref{lemma:simulation-2nd-order}.5.)}
\end{array}
\]

Case $\pi$: Sub-case $E=B::l'$.
$$
\begin{array}{rcll}
(\ctt{tl}(B::l'):G,K)&=&(tl:\suc{G},\lambda m.(l':G,K)(mB^*))&\\
&\rightarrow^+_{\beta}&(t\,(\app l{(B::l')}):G,K)&\textrm{(Lemma
    \ref{lemma:simulation-2nd-order}.10.)}\rlap{\hbox to21 pt{\hfill}\qEd}\\
\end{array}
$$

%---------------------------------
\begin{cor}\label{cor:SN-for-LambdaJmseS}
The typable terms of $\lJmseS$ are strongly normalising.
\end{cor}
%---------------------------------
A technically more involved CGPS for the Church-style version of
$\lJmseS$ into Church-style system~$F$ can be given along the lines of
\cite{Matthes06}, where the colon translation has to be made relative
to a context $\Gamma$.

%++++++++++++++++++++++++++++++++++++++++++++++++++++++++++++++++++++++
\subsection{$F^{\omega}$ and Higher-Order Logic}\label{subsec:Fomega-and-higher-order}

In the second order systems, one assumes that $X$ in the
quantification $\forall X.A$ ranges over the domain $\PROP$ of all
propositions (or types). In this subsection we study systems
allowing the formation of other \emph{domains of quantification},
usually denoted $\D$, $\E$. Quantification now has the form
$\forall X:\D.A$, but, at the proof-term level, abstraction
$\Lambda X.t$ remains domain-free.

In the following we formulate intuitionistic higher-order
predicate logic, both in the natural deduction format $\HOL$, and
the sequent calculus format $\HJmse$. A minor restriction in each
of these systems gives two formulations ($\ldfO$ and $\lJmseO$,
respectively) of intuitionistic higher-order propositional logic,
or system $F^{\omega}$.

\subsubsection{Domains of quantification}
Domains (of quantification)
are given by:

$$
\begin{array}{rcl}
\D,\E & ::= & \PROP\,|\,\X\,|\,\D\to\D
\end{array}
$$

\noindent $\X$ ranges over a set of \emph{domain variables}. These
play the role of ``sorts'' in multi-sorted first-order logic. The
set of domains is very much like Church's structure of simple types,
except that, besides $\PROP$ (the type of propositions), Church
only admitted one other base type $\iota$ of \emph{individuals}.

Next come the propositional, or type, or individual,
function(al)s:

$$
\begin{array}{rcl}
A,B,C & ::= & X\,|\,\lbd X.A\,|\,AB\,|\,A\imp B\,|\,\forall X:\D.A
\end{array}
$$

\noindent These are the inhabitants of domains. $X$ ranges over a
set, whose elements may be seen as \emph{type variables}, or
\emph{propositional variables}, or \emph{individual variables},
etc. In the last case a meta-variable like $\mathsf{x}$ would be
more expressive. Also, one may employ meta-variables $\varphi$ and
$\psi$ instead of $A$, if one wants to emphasize that the
inhabitant is a proposition, or $\mathsf{t}$ if one wants to
emphasize that the inhabitant lives in some domain of individuals
$\X$. $\lbd X.A$ and $AB$ are the generic, and usual, mechanism
for building inhabitants at all levels of the domain
structure.\footnote{Notation: different forms of abstraction are
denoted by variants of the symbol $\lambda$, but application is
always denoted by juxtaposition.}

The relationship between domains and their inhabitants is governed
by \emph{domain assignment rules}. Let $\Delta$ range over
consistent sets of declarations $X:\D$. Such rules derive sequents
of the form $\Delta\vdash A:\D$, as described in
Figure~\ref{fig:domains}.
\begin{figure}\caption{Domain assignment rules for higher-order logic}\label{fig:domains}
$$
\begin{array}{c}
\infer[Ax]{\Delta\vdash X:\D}{(X:\D)\in\Delta}\\ \\
\infer[I\to]{\Delta\vdash\lbd X.A:\D\to\E}{\Delta,X:\D\vdash
A:\E}\qquad\infer[E\to]{\Delta\vdash AB:\E}{\Delta\vdash
A:\D\to\E&\Delta\vdash B:\D}\\ \\\infer[F\imp]{\Delta\vdash A\imp
B:\PROP}{\Delta\vdash A:\PROP&\Delta\vdash
B:\PROP}\qquad\infer[F\forall]{\Delta\vdash\forall
X:\D.A:\PROP}{\Delta,X:\D\vdash A:\PROP}
\end{array}
$$
\end{figure}
\noindent Besides the ordinary rules of the simply-typed
$\lambda$-calculus, one has two \emph{formation} rules. If
$\Delta\vdash A:\PROP$, then we say that $A$ is a
\emph{$\Delta$-proposition}, or just \emph{proposition}.
Alternative terminology is ``formula'' or ``type''.

Finally, the inhabitants of domains may reduce according to the
following reduction rule:

$$
\begin{array}{rrcl}
(\beta_0) & (\lbd X.A)B& \rightarrow & [B/X]A\enspace.
\end{array}
$$

The given definition of domains, their inhabitants, and the
derivable sequents $\Delta\vdash A:\D$ remains fixed for the rest
of this subsection (that is, in all the systems $\ldfO$, $\HOL$,
$\lJmseO$, and $\HJmse$), except for one thing: in $\ldfO$ and
$\lJmseO$, domain variables $\X$ are not allowed.

\subsubsection{Systems $\ldfO$ and $\HOL$}
We now define the natural
deduction system $\HOL$ and its minor variant $\ldfO$.
Specifically, we define proof expressions and their ``typing''
rules, that is, the rules governing what expressions inhabit what
propositions. At this level, the systems $\ldfO$ and $\HOL$ are
indistinguishable; indeed, the single difference is the one already
pointed out at the domains level.

In addition, at this level, also $\HOL$ and $\ldfS$ would be
indistinguishable, provided that (i) we had defined $\ldfS$ with a
trivial domain level $\D=\PROP$, and with formal ``domain
assignment rules'' generating the types/ propositions; (ii) we
wrote $\forall X:\D.A$ and not just $\forall X.A$; (iii) sequents
$\Gamma\vdash t:A$ carried an outer set $\Delta$ declaring
necessary variables $X$ with domain $\PROP$. So, what follows
may be used as a recapitulation of $\ldfS$.

In $\HOL$ one has the following \emph{proof terms}:

$$
\begin{array}{rcl}
t,u,v & ::= & x\,|\,\lambda x.t\,|\,tu\,|\,\Lambda X.t\,|\,tA
\end{array}
$$

Proof terms are assigned to propositions through \emph{proposition
assignment rules}, which generate sequents of the form
$\Delta;\Gamma\vdash t:A$ according to the rules of Figure~\ref{fig:HOL}. Here $\Gamma$ is a consistent set of
declarations $x:A$; in addition we expect
$\Delta\vdash\Gamma:\PROP$ and $\Delta\vdash A:\PROP$, whenever
$\Delta;\Gamma\vdash t:A$ is generated. The notation
$\Delta\vdash\Gamma:\PROP$ means
$(x:A)\in\Gamma\Rightarrow\Delta\vdash A:\PROP$.
\begin{figure}\caption{Proposition assignment rules of $\HOL$}\label{fig:HOL}
$$
\begin{array}{c}
\infer[Ax]{\Delta;\Gamma\vdash x:A}{\Delta\vdash\Gamma:\PROP & (x:A)\in\Gamma}\\ \\
\infer[E\imp]{\Delta;\Gamma\vdash tu:B}{\Delta;\Gamma\vdash
t:A\imp B&\Delta;\Gamma\vdash u:A}\quad\quad
\infer[I\imp]{\Delta;\Gamma\vdash\lambda x.t:A\imp
B}{\Delta;\Gamma,x:A\vdash
t:B}\\
\\
\infer[E\forall]{\Delta;\Gamma\vdash
tB:[B/X]A}{\Delta;\Gamma\vdash t:\forall X:\D.A&\Delta\vdash
B:\D}\quad\quad \infer[I\forall]{\Delta;\Gamma\vdash\Lambda
X.t:\forall X:\D.A }{\Delta,X:\D;\Gamma\vdash
t:A}\\
\\
\infer[Conv.]{\Delta;\Gamma\vdash t:B}{\Delta;\Gamma\vdash t:A &
\Delta\vdash B:\PROP & A=_{\beta_0}B}

\end{array}
$$
\end{figure}
Proof terms reduce according to these two reduction rules:

$$
\begin{array}{rrcl}
(\beta_1) & (\lambda x.t)u& \rightarrow & [u/x]t\\
(\beta_2) & (\Lambda X.t)B& \rightarrow & [B/X]t\\
\end{array}
$$

\noindent Proof terms are capable of $\beta_0$-reduction, via the
closure rule $B\to_{\beta_0}B'\Longrightarrow tB\to_{\beta_0}tB'$.

\subsubsection{Systems $\lJmseO$ and $\HJmse$}
We now define the sequent
calculi $\HJmse$ and its minor variant $\lJmseO$. Again, at the
level of proof expressions and their ``typing'' rules, the systems
$\lJmseO$ and $\HJmse$ are indistinguishable; indeed, the single
difference is the one already pointed out at the domains level.

In addition, at this level, also $\HJmse$ and $\lJmseS$ would be
indistinguishable, under the same provisos as before for the
indistinguishability of $\HOL$ and $\ldfS$. Hence, also the
following definition is mostly a recapitulation of $\lJmseS$.

In $\HJmse$ one has the following \emph{proof expressions}:

$$
\begin{array}{lcrcl}
\textrm{(Proof terms)} &  & t,u,v & ::= & x\,|\,\lambda x.t\,|\,\Lambda X.t\,|\,\ctt c\\
\textrm{(Proof co-terms)} &  & l & ::= & []\,|\,u::l\,|\,B::l\,|\,(x)c\\
\textrm{(Proof commands)} &  & c & ::= & tl
\end{array}
$$

Proposition assignment rules generate sequents of the forms
$\Delta;\Gamma\vdash t:A$, and $\seql{\Delta;\Gamma}{l}{B}{A}$,
and $\seqc{\Delta;\Gamma}{c}{B}$. In these sequents we expect
$\Delta\vdash \Gamma:\PROP$, $\Delta\vdash A:\PROP$, and
$\Delta\vdash B:\PROP$. The rules are shown in Figure~\ref{fig:HJmse}.
\begin{figure}\caption{Proposition assignment rules of $\HJmse$}\label{fig:HJmse}
$$
\begin{array}{c}
\infer[LAx]{\seql{\Delta;\Gamma}{[]}{A}{A}}{\Delta\vdash\Gamma:\PROP & \Delta\vdash A:\PROP}
\quad\quad\infer[RAx]{\Delta;\Gamma\vdash x:A}{\Delta\vdash\Gamma:\PROP & (x:A)\in\Gamma}\\ \\
\infer[L\imp]{\seql{\Delta;\Gamma}{u::l}{A\imp B}{
C}}{\Delta;\Gamma\vdash u:A&
\seql{\Delta;\Gamma}{l}{B}{C}}\quad\quad
\infer[R\imp]{\Delta;\Gamma\vdash\lambda x.t:A\imp
B}{\Delta;\Gamma,x:A\vdash
t:B}\\
\\
\infer[L\forall]{\seql{\Delta;\Gamma}{B::l}{\forall X:\D.A}{
C}}{\Delta\vdash B:\D&
\seql{\Delta;\Gamma}{l}{[B/X]A}{C}}\quad\quad
\infer[R\forall(X\notin\Gamma)]{\Delta;\Gamma\vdash\Lambda
X.t:\forall X:\D.A }{\Delta,X:\D;\Gamma\vdash
t:A}\\
\\
\infer[LSel]{\seql{\Delta;\Gamma}{(x)c}{A}{B}}{\seqc{\Delta;\Gamma,x:A}{c}{B}}
\quad\quad \infer[RSel]{\Delta;\Gamma\vdash
\ctt c:A}{\seqc{\Delta;\Gamma}{c}{A}}\\
\\
\infer[Cut]{\seqc{\Delta;\Gamma}{tl}{B}}{\Delta;\Gamma\vdash t:A &
\seql{\Delta;\Gamma}{l}{A}{B}}\\
\\
\infer[Conv.]{\Delta;\Gamma\vdash t:B}{\Delta;\Gamma\vdash t:A &
\Delta\vdash B:\PROP & A=_{\beta_0}B}
\end{array}
$$
\end{figure}

The rules for the reduction of proof expressions are:

$$
\begin{array}{rrcl}
(\beta_1) & (\lambda x.t)(u::l)& \rightarrow & u((x)tl)\\
(\beta_2) & (\Lambda X.t)(B::l)& \rightarrow & ([B/X]t)l\\
(\pi) & \ctt{tl}E & \rightarrow & t\,(\app lE)\\
(\sigma) & t(x)c & \rightarrow & [t/x]c\\
(\mu) & (x)xl & \rightarrow & l,\textrm{ if $x\notin
l$}\\
(\epsilon) & \ctt{t[]} & \rightarrow & t
\end{array}
$$

\noindent Proof expressions are capable of $\beta_0$-reduction,
via the closure rule $B\to_{\beta_0}B'\Longrightarrow
B::l\to_{\beta_0}B'::l$.

\subsubsection{CGPS translations}

We will see that, when the CGPS translation is extended to $\lJmseO$
and $\HJmse$, its properties (type soundness and the simulation
theorem) remain valid and are proved almost \emph{verbatim} relative
to the second-order case. Here is an explanation. The proofs of the
properties of the CGPS translation have two components. The first
component is a proof that the CGPS translation behaves well
relative to domain inhabitants/assignment. This comprises (i)
domain soundness (Lemma \ref{lemma:domain-assignment-to-CGPS}
below); (ii) commutation with substitution of type variables $X$;
(iii) simulation of $\beta_0$ (Lemma \ref{lem:simulation-beta0}).
This component depends on the domain inhabitants/assignment and
inhabitants reduction ($\beta_0$) of the system where the
translation is defined. Very little variation exists between
$\lJmseS$, $\lJmseO$, and $\HJmse$ regarding these aspects, the only
singularity being that there is no $\beta_0$ at second order. The
second component is the proper proofs of type soundness and strict
simulation, which are all the same for $\lJmseS$, $\lJmseO$, and
$\HJmse$, except for one inductive case of the strict simulation
theorem, absent in $\lJmseS$, and relative to $\beta_0$ reduction at
proof-expression level.

We define a CGPS translation from $\HJmse$ to $\HOL$. It can be
seen as a CGPS translation from $\lJmseO$ to $\ldfO$ as well, and
generalises only slightly the previous CGPS translation from
$\lJmseS$ to $\ldfS$, by providing translations for $\lbd X.A$ and
$AB$.

Domains remain fixed, but their inhabitants are translated as in Figure~\ref{fig:prop-trans}.
\begin{figure}\caption{Translation of propositional/individual function(al)s}\label{fig:prop-trans}
\[
\begin{array}{rcl}
X^*&=&X\\
(A\imp B)^*&=&\neg\ov{B}\imp\neg\ov{A}\\
(\forall X:\D.A)^*&=&\forall X:\D.\ov{A}\\
(\lbd X.A)^*&=&\lbd X.A^*\\
(AB)^*&=&A^*B^*\\
&&\\
\ov{A}&=&\top\imp\neg\neg A^*\\
\end{array}
\]
\end{figure}
Recall that the relation of domain assignment of $\HJmse$ is the
same as that of $\HOL$. Such relation is intended in the following
result.

%-----------------------------------------
\begin{lem}[Domain soundness]\label{lemma:domain-assignment-to-CGPS}The following holds:\\
$$
\infer{\Delta\vdash A^*:\D}{\Delta\vdash
A:\D}\qquad\infer{\Delta\vdash \ov{A}:\PROP}{\Delta\vdash A:\PROP}
$$
\end{lem}
%-----------------------------------------
\begin{proof}
By simultaneous induction on $A$.
\end{proof}

Recall also that the relation of domain assignment of $\lJmseO$ is
the same as that of $\ldfO$. If this latter relation is intended,
the previous result also holds, with the same proof. The previous
lemma generalises the fact that, at second order, if $A$ is a
proposition (type), then so is $A^*$ and $\ov{A}$.

The same grammar generates the sets of proof expressions of $\HJmse$,
$\lJmseO$, and $\lJmseS$; another single grammar generates the sets
of proof expressions of $\HOL$, $\ldfO$, and $\ldfS$. These two
grammars are already known from the second-order systems, so the CGPS
translation at the level of proof expressions is known and we do not
repeat it.

The equations $([B/X]A)^*=[B^*/X]A^*$ and
$\ov{[B/X]A}=[B^*/X]\ov{A}$ still hold, and are proved by the same
simultaneous induction, supplemented with the straightforward new
cases $\lbd X.A$ and $AB$.

%---------------------------------------------
\begin{lem}
\label{lem:simulation-beta0} If $A\rightarrow_{\beta_0} B$ in
$\HJmse$ (resp. $\lJmseO$), then $A^*\rightarrow_{\beta_0} B^*$
and $\ov{A}\rightarrow_{\beta_0} \ov{B}$ in $\HOL$ (resp.
$\ldfO$).
\end{lem}
%----------------------------------------------
\begin{proof}
Straightforward induction on $A\rightarrow_{\beta_0} B$. The base
case follows from $([B/X]A)^*=[B^*/X]A^*$. The inductive cases are
routine.
\end{proof}

Then one obtains the admissible typing rules of
Figure~\ref{fig:HJmseTypes}.
\begin{figure}\caption{Admissible proposition assignment rules for CGPS translation of $\HJmse$}\label{fig:HJmseTypes}
$$
\begin{array}{c}
\infer{\Delta;\ov{\Gamma}\vdash\ov{t}:\ov{A}}{\Delta;\Gamma\vdash t:A}\\ \\
\infer{\Delta;\ov{\Gamma},\Gamma'\vdash
(t:G,K):\bot}{\Delta;\Gamma\vdash
t:A&\Delta;\ov{\Gamma},\Gamma'\vdash G:\top&\Delta;\ov{\Gamma},\Gamma'\vdash K:\neg A^*}\\ \\
\infer{\Delta;\ov{\Gamma},\Gamma'\vdash
(l:G,K):\neg\ov{A}}{\Delta;\Gamma|l:A\vdash
B&\Delta;\ov{\Gamma},\Gamma'\vdash G:\top&\Delta;\ov{\Gamma},\Gamma'\vdash K:\neg B^*}\\ \\
\infer{\Delta;\ov{\Gamma},\Gamma'\vdash
(c:G,K):\bot}{\seqc{\Delta;\Gamma}{c}{A}&\Delta;\ov{\Gamma},\Gamma'\vdash
G:\top&\Delta;\ov{\Gamma},\Gamma'\vdash K:\neg A^*}
\end{array}
$$
\end{figure}
This is the same typing obeyed by the CGPS translation
from $\lJmseS$ to $\ldfS$, provided, as remarked before, $\lJmseS$
and $\ldfS$ are defined with a formal level of domains, etc.

%------------------------------------------------
\begin{lem}\label{lemma:simulation-HOL} The CGPS translations of $\HJmse$ into $\HOL$,
and of $\lJmseO$ into $\ldfO$, satisfy the items (1) to (10) of
Lemma \ref{lemma:simulation-2nd-order}.
\end{lem}
%-----------------------------------------------

%---------------------------------------------
\begin{thm}[Simulation]
\label{theorem:simulation-cgpsHOL} If $t\rightarrow u$ in $\HJmse$
(resp. $\lJmseO$), then $\ov{t}\rightarrow^+_{\beta} \ov{u}$ in
$\HOL$ (resp. $\ldfO$).
\end{thm}
%----------------------------------------------

\begin{proof}
The same proof as in the second-order case applies. There is only
one new inductive case, to prove
$(B::l:G,K)\to_{\beta}^+(B'::l:G,K)$, when $B\to_{\beta_0}B'$, a
case which is an immediate consequence of Lemma
\ref{lem:simulation-beta0} and the definition of the CGPS
translation.
\end{proof}

%---------------------------------
\begin{cor}\label{cor:SN-for-HOL}
The typable terms of $\HJmse$ and $\lJmseO$ are strongly
normalising.
\end{cor}

\section{Further remarks}\label{sec:conclusions}

\textbf{Contributions.} This article provides reduction-preserving
CGPS translations of $\lJmse$ and other intuitionistic calculi,
hence obtaining embeddings into the simply-typed
$\lambda$-calcu\-lus and proving strong normalisation. As a
by-product, the connections between systems like $\lJ$ and $\lJm$
and the intuitionistic fragment of $\lmmt$ are detailed, and
confluence for them obtained. It is shown that all the results
smoothly extend to systems with quantification over propositions
and even functionals over propositions and (many-sorted)
individuals. In all cases, the sequent-calculus format is embedded
into the natural-deduction style.

\textbf{C(G)PS and strong normalisation.} In the literature one
finds strong normalisation proofs for sequent calculi
\cite{Dragalin88,DyckhoffUrbanJLC,LengrandWRS2003,LengrandEtAl06,PolonovskiFOSSACS2004,UrbanBiermanTLCA1999},
but not by means of CPS translations; or CPS translations for
natural deduction systems
\cite{BartheHatcliffSoerensen97,BartheHatcliffSoerensen99,Groote02,HarperLillibridgeJFP96,IkedaNakazawa06,NakazawaTatsuta06}.

This article provides, in particular, a reduction-preserving CGPS
translation for the lambda-calculus with generalised applications
$\lJ$. \cite{NakazawaTatsuta06} covers full propositional classical
logic with general elimination rules and its intuitionistic
implicational fragment corresponds to $\lJ$. However,
\cite{NakazawaTatsuta06} does not prove a strict simulation by CPS
(permutative conversions are collapsed), so an auxiliary argument in
the style of de Groote \cite{Groote02}, involving a proof in isolation
of SN for permutative conversions, is used.

In Curien and Herbelin's work
\cite{CurienHerbelinFP2000,HerbelinHabilitation} one finds a CPS
translation $(\_)^n$ of the call-by-name restriction of $\lmmt$.
We compare $(\_)^n$ with our $\ov{(\_)}$. (i) $(\_)^n$ generalises
Hofmann-Streicher translation \cite{HofmannStreicherJournal};
$\ov{(\_)}$ generalises Plotkin's call-by-name CPS translation
\cite{PlotkinTLCA1975}. (ii) $(\_)^n$ does not employ the colon
operator; $\ov{(\_)}$ does employ (we suspect that doing
administrative reductions at compile time is necessary to achieve
strict simulation of reduction); (iii) $(\_)^n$ is defined for
expressions where every occurrence of $u::l$ is of the particular
form $u::E$; no such restriction is imposed in the definition of
$\ov{(\_)}$. (iv) at some points it is unclear what the properties
of $(\_)^n$ are, but no proof of strong normalisation is claimed;
the CGPS $\ov{(\_)}$ strictly simulates reduction and thus achieves a proof
of strong normalisation.

\textbf{Higher-order sequent calculi.} Our formulations of system
$F^{\omega}$ and higher-order logic in the sequent calculus format
were helpful for showing the wide applicability of the CGPS
technique. Nevertheless, they are another experience in the
formulation of type theories as sequent calculi
\cite{LengrandEtAl06}. We adopted the guideline that only
proof-expression could suffer a change in the proof-theoretical
format, but other, more ``uniform'', possibilities exist, where also
the domain assignment relation is changed to the sequent calculus
format. An improvement, in view of proof-search, is to restrict the
conversion rule of the typing system to an expansion rule
\cite{SeldinTCS2000}. Finally, in $\HJmse$, $\lJmseO$, and $\lJmseS$
we re-encounter explicit substitutions in higher-order type theories
\cite{BlooMSCS2001,MunozMSCS2001}, but with a simpler treatment (no
explicit execution) and in a simpler setting (no dependent types).

\textbf{Future work.} We plan to extend the technique of
continuation-and-garbage passing to $\lmmt$ and to
dependently-typed systems. We tried to extend the CGPS to CBN
$\lmmt$, but already for a CPS translation, we do not see how to
profit from the continuation argument for the translation of
co-terms and commands. Moreover, a special case of the rule we
call $\pi$ corresponds to the renaming rule $a(\mu b.M)
\rightarrow [a/b]M$ of $\lb\mu$-calculus. This rule is evidently
not respected by the CGPS translation by Ikeda and Nakazawa
\cite{IkedaNakazawa06} (nor by the CPS they recall) since the
continuation argument $K$ is omitted in the interpretation of the
left-hand side but not in the right-hand side. So, new ideas or
new restrictions will be needed.\\

\noindent \emph{Acknowledgements:} We thank the referees of this
journal version, whose reviews helped to improve considerably the
first submission. We also thank the referees of the conference
version \cite{ESMP:TLCA07} for, among other things, pointing out the
work of Nakazawa and Tatsuta, whom we thank for an advanced copy of
\cite{NakazawaTatsuta06}. The first and third authors are supported
by FCT through the Centro de Matem\'atica da Universidade do Minho.
The second author thanks for an invitation by that institution to
Braga in October 2006 and in May 2007. All authors were also
supported by the European Union FP6-2002-IST-C Coordination Action
510996 ``Types for Proofs and Programs''.

%\bibliography{LambdaJmse}
%\bibliographystyle{plain}

\end{document}